\documentclass[format=acmsmall]{acmart}
\pdfoutput=1 

\setcopyright{none} 
\renewcommand\footnotetextcopyrightpermission[1]{}
\usepackage{gensymb}
\usepackage{graphicx}
\usepackage{url}
\usepackage{amsmath}
\usepackage{float}
\usepackage{amsfonts}
\usepackage{latexsym,epsfig,color,verbatim}
\usepackage{algorithm}
\usepackage{algpseudocode}
\usepackage{pifont}
\usepackage{lipsum}
\usepackage{mwe}
\usepackage{slashbox}
\usepackage{color}
\usepackage{cancel}
\usepackage{multirow}
\usepackage{subcaption}

\newcommand{\BEQA}{\begin{eqnarray}}
\newcommand{\EEQA}{\end{eqnarray}}

\newtheorem{lemma}{Lemma}
\newtheorem{theorem}{Theorem}

\newtheorem{remark}{Remark}

\usepackage{bbm}
\usepackage{slashbox}

\def\beq{\begin{equation}}
\def\eeq{\end{equation}}
\def\beqn{\begin{equation*}}
\def\eeqn{\end{equation*}}
\def\bearn{\begin{eqnarray*}}
\def\eearn{\end{eqnarray*}}
\def\bear{\begin{eqnarray}}
\def\eear{\end{eqnarray}}
\def\barr{\begin{array}}
\def\earr{\end{array}}

\newcommand {\done} {\quad\vrule height4pt WIDTH4PT}
\newcommand{\E}{\mathbb{E}}
\newcommand{\N}{\mathbb{N}}
\newcommand{\Prob}{\mathbb{P}}
\newcommand{\R}{\mathbb{R}}
\newcommand{\Z}{\mathbb{Z}}

\def\Z{\mathbb{Z}}

\usepackage{soul}
\usepackage{mathtools}

\begin{document}
\title{A New Upper Bound on Cache Hit Probability for Non-anticipative Caching Policies}

\author{Nitish K. Panigrahy}
\affiliation{
  \institution{University of Massachusetts Amherst }
  \state{Amherst, MA 01003} 
  \country{USA}
}
\email{nitish@cs.umass.edu}

\author{Philippe Nain}
\affiliation{
  \institution{Inria Grenoble -- Rh\^one-Alpes}
  \state{Lyon 69364} 
  \country{France}
}
\email{philippe.nain@inria.fr}

\author{Giovanni Neglia}
\affiliation{
  \institution{Inria Sophia Antipolis -- M\'editerran\'ee, UCA}
  \country{France}
}
\email{giovanni.neglia@inria.fr}

\author{Don Towsley}
\affiliation{
  \institution{University of Massachusetts Amherst }
  \state{Amherst, MA 01003} 
  \country{USA}
}
\email{towsley@cs.umass.edu}

\maketitle


Caching systems have long been crucial for improving the performance of a wide variety of network and web based online applications. In such systems, end-to-end application performance heavily depends on the fraction of objects transfered from the cache, also known as the \emph{cache hit probability}. Many caching policies have been proposed and implemented to improve the hit probability. In this work, we propose a new method to compute an upper bound on hit probability for all non-anticipative caching policies, i.e., for policies that have no knowledge of future requests. Our key insight is to order the objects according to the ratio of their Hazard Rate (HR) function values to their sizes and place in the cache the objects with the largest ratios till the cache capacity is exhausted. 
Under some statistical assumptions, we prove that our proposed HR to size ratio based ordering model computes the maximum achievable hit probability and serves as an upper bound for all non-anticipative caching policies. We derive closed form expressions for the upper bound under some specific object request arrival processes. We also provide simulation results to validate its correctness and to compare it to the state-of-the-art upper bounds. We find it to be tighter than state-of-the-art  upper bounds for a variety of  object request arrival processes. 

\section{Introduction}\label{sec:intro}

Caches are pervasive in computing systems, and their importance is reflected in many networks and distributed environments including \emph{content delivery networks} (CDNs). In such networks, the end user quality of experience primarily depends on whether the requested object is cached near the user. Thus the cache \emph{hit probability}, i.e., the percentage of requests satisfied by the cache, plays an important role in determining the end-to-end application performance. In general, the number of objects available in a system is much larger than the cache capacity. Hence, the design of caching algorithms typically focuses on maximizing the overall cache hit probability.  Also, maximizing the cache hit probability corresponds to minimizing the expected retrieval time, the load on the server and on the network when object sizes are equal. 

One possible way to improve cache hit probability is by increasing cache capacity. However, increasing cache capacity only logarithmically improves cache hit probability\cite{Breslau1999, Almeida1996}. Thus improving caching policies seems to be more effective for maximizing the overall cache hit probability.
In practice, most caches employ  \emph{least-recently used} (LRU) or its variants often coupled with call admission or prefetching \cite{berger18}. Apart from LRU, other well known eviction policies include  LFU, FIFO, RANDOM. There has been plethora of work \cite{Jiang2002, Megiddo2003, Tanenbaum01, Arlitt2000, Beckmann2018, Cao1997, Jaleel2010} on improving cache hit probabilities in the literature. 
In order to gauge the potential effectiveness of these eviction policies, an upper bound on maximum achievable cache hit probability for a given cache capacity has been widely adopted \cite{Arlitt2000}.


\subsection{Offline upper bound}
\label{ssec:offline-ub}
For equal size object, B\'el\'ady's algorithm or MIN \cite{aho71}  has been widely used as an upper bound for cache hit probability among all feasible on demand and online caching policies, togetherly known as \emph{non-anticipative} policies. However, B\'el\'ady's algorithm is an offline algorithm, i.e., it assumes exact knowledge of future requests. Offline upper bounds on object hit probability have been proposed for variable (different) size object  \cite{berger18}. Often system designers do not have access to the exact request trace, but can estimate the statistical properties of the object request process such as the \emph{inter-request time} (irt) distribution. Also, caching studies typically include model driven simulations. Thus the following natural question arises. {\it With limited knowledge of the object arrival process and no look ahead option, can we provide an upper bound on the cache hit probability for any feasible online caching policy?}

\subsection{Our Approach: Hazard Rate based upper bound}
\label{ssec:IRM}
When object requests follow the {\it Independent Reference Model} (IRM), i.e., when objects are referenced independently with fixed probabilities, ideal {Least-Frequently Used} (LFU) caching policy is asymptotically optimal in terms of object hit probability. However, general request processes are more complex and correlated.
%

In this work, we assume a larger class of statistical models for object reference streams, see Section \ref{ssec:number-hits} and Section \ref{sub:hrb} for more details. 
We also assume that the {\it hazard rate} (HR) function (or conditional intensity) associated with this point process is well defined and can be computed at all points of time $t$. Here, the HR function is the conditional density of the occurrence of an  object request at time $t$, given the realization of the request process over the interval $[0, t)$ \cite{Daley2003}. 

We now propose the HR based upper bound as follows.  When objects have equal size, at any time $t$ we determine the HR values of each object and  place in the cache the $B$ objects which have the largest HR values. When objects have different sizes, we sort the objects according to the ratio of their HR values at time $t$ to their sizes in decreasing order. Note that, an  ideal LFU policy keeps track of number of times an object is referenced and order them accordingly in the cache. Similarly, in our upper bound, we keep an ordered list but on the basis of ratio of HR values to object sizes. We then place in the cache the objects with the largest ratios till the cache capacity is exhausted. We emphasize that we do not provide new caching policies that would outperform all other policies but instead we provide various upper bounds on the cache hit probability. 


Our contributions are summarized below:
\begin{enumerate}
\item We present a new upper bound for cache hit probability among all non-anticipative caching policies:
\begin{itemize}
\item When objects have equal sizes, a simple HR based ordering for the objects provides an upper bound on cache hit probability.
\item For variable size objects, we order the objects with respect to the ratio of their HR function values to their objects sizes and provide upper bounds on the  byte and object hit probabilities. 
\end{itemize}
\item We derive closed form expressions for the upper bound under some specific object request arrival processes.
\item We evaluate and compare the HR based upper bound with different cache replacement policies for both synthetic and real world traces.
\end{enumerate}

The rest of this paper is organized as follows. In Section \ref{sec:model} we formally present the HR based upper bound for equal size objects. In Section \ref{sec:vars} we develop HR based upper bound for variable size objects. We discuss the HR based upper bound for specific object request arrival processes in Section \ref{sec:disc}. We perform simulation experiments to compare HR based upper bound with other policies in Section \ref{sec:perf}. 
Finally, the conclusion of this work and potential future works are given in Section \ref{sec:con}.

\section{Equal Size Objects}\label{sec:model}

We consider a cache of capacity $B$ serving $n$ distinct equal size objects. Without loss of generality we assume that all objects have size one. Later in Section \ref{sec:vars}, we also consider objects with different sizes. Let $\mathcal{D}=\{1, \cdots, n\}$ be the set of objects.  We denote by $\N=\{0,1,2,\ldots\}$ the set of all nonnegative integers, $\N^*=\N-\{0\}$, and by 
$\Z=\{0,\pm 1,\pm 2, \ldots\}$ the set of all integers.

\subsection{Number of Hits for General object Arrival Processes}
\label{ssec:number-hits}

Let $\{0<T_{i,1} <T_{i,2}<\cdots\}$ be the successive time epochs when object $i$ is requested.  Assume $\{T_{i,k}\}_k$ is a regular point process \cite{Daley2003}. Define $X_{i,k}=T_{i,k}- T_{i,k-1}$  for $k\geq 2$  and $X_{i,1}=T_{i,1}$. 
Let $\{0<T_1<T_2<\cdots\}$ be the point process resulting from the superposition of the point processes $\{T_{i,k}\}_k$, $i=1,\ldots,n$.  
For $t>0$, define  $\mathcal{H}_{i,t}=\{T_{i,k}, k\in\N^*: T_{i,k}<t\}$ the history of the point process $\{T_{i,k}\}_k$ up to time $t$ and 
 $\mathcal{H}_t=\cup_{i=1}^n \mathcal{H}_{i,t}$ the entire history of all point processes $\{T_{1,k} \}_k,\ldots,\{T_{n,k}\}_k$ up to time $t$.  Notice that
 $\mathcal{H}_{i,t}$ is left-continuous for all $i$ and so is $\mathcal{H}_t$. In particular, $T_k\not\in \mathcal{H}_{T_k}$ for all $k$.
Define $k_i(t)=\max\{k\geq 1 : T_{i,k-1}<t\}$, so that exactly $k_i(t)-1$ requests for object $i$ have been made in $[0,t)$.

Assume that the request object processes $\{T_{i,k}\}_k$, $i=1,\ldots,n$, are conditionally independent $\forall t>0$, in the sense that
\begin{align}
\label{assumption-arrivals}
\Prob(T_{1,k_1(t)}\geq\, t_1, \cdots, T_{n,k_n(t)}\geq\, t_n\,|\, \mathcal{H}_t)= \prod_{i=1}^n  \Prob(T_{i,k_i(t)}\geq\,t_i \,|\, \mathcal{H}_{i,t}),
\end{align}
for all $t_1\geq t, \ldots, t_n\geq t$.

Given $T_{i,k}=t_{i,k}$ for $k\geq 1$, the hazard rate function of  $\{T_{i,k}\}_k$ at time $t$ is defined by the piecewise function \cite[Definition 7.2.II, p. 231]{Daley2003}
\begin{align}
\lambda^*_i(t)=\left\{\begin{array}{ll}
\frac{ \frac{d}{dt}P(X_{i,1}<t )}{P(X_{i,1}>t)} &\mbox{for $0<t\leq t_{i,1}$,}\\
\frac{\frac{d}{dt}P(X_{i,k}<t-t_{i,k-1} \, | \, T_{i,j}=t_{i,j}, j\leq k-1)}{P(X_{i,k}>t-t_{i,k-1}\,| \, T_{i,j}=t_{i,j}, j\leq k-1)}  &\mbox{for $t_{i,k-1}<t\leq t_{i,k}$, $k\geq 2$.}\\
                                   \end{array}
                         \right.
\label{hazard-rate}
\end{align}
In (\ref{hazard-rate}) the existence of $\frac{d}{dt}P(X_{i,1}<t )$ and    $\frac{d}{dt}P(X_{i,k}<t-t_{i,k-1}\, |\,T_{i,j}=t_{i,j}, j\leq k-1)$  for $k\geq 2$, follows from the assumption  that  $\{T_{i,k}\}_k$ is a regular point process  \cite{Daley2003}.

We assume that the cache is empty at time $t=0$ to avoid unnecessary  notational complexity but
all results in the paper hold without this assumption as long as the state of the cache is known at time $t=0$.
A  caching policy $\pi$ determines at any time $t$ which $B$ objects among the $n$ available objects are cached. Formally, $\pi$ is a measurable deterministic mapping from
$\R\times (\R\times \R)^\infty\to  S_B$, 
 where $S_B$ is the set of subsets of $\{1,\ldots,n\}$ which contain $B$ elements. In this setting, $\pi(t,\mathcal{H}_t)$ gives the $B$ objects that are  cached at time $t$ based on the knowledge of the overall request process up to $t$.
Let $\Pi$ be the collection of all such policies. Note that policies in $\Pi$ are  non-anticipative policies, in the sense that they do not know when future requests will occur.

We will only consider deterministic policies although the setting can easily be extended  to random policies (in this case $\pi:\R\times (\R\times \R)^\infty\to  \mathcal{Q}(S_B)$ where
$\mathcal{Q}(S_B)$ is the set of probability distributions on  $S_B$). 
 

We introduce the {\em  hazard rate} (HR) based  {\em rule}  for equal-size objects, abbreviated as HR-E.
At any time $t$ and given $\mathcal{H}_t$, HR-E (i) determines the hazard rate function of each object and  (ii) places in the cache the $B$ documents which have the largest hazard rate functions, {\em i.e.} if  $\lambda^*_{i_1}(t)\geq \cdots\geq\lambda^*_{i_n}(t)$ then objects $i_1,\ldots,i_B$ are cached at time $t$ (ties between equal rates are broken randomly). 
We call it a rule, not a policy and will use it as a way to upper-bound various performance metrics (see next), which is the goal of this paper.

Let $B^\pi_k\in S_B$ be the state of the cache just before time $T_k$ (sometimes abbreviated to $T_k-$) under $\pi$. 
Call $R_k$ the object requested at time $T_k$ under $\pi$ and define 
\begin{equation}
H^\pi_k={\bf 1}(R_k\in B^\pi_k),
\label{def-Hk}
\end{equation}
{\em i.e.}, $H^\pi_k=1$ if the $k$-th requested object  is in the cache and 
$H^\pi_k=0$ otherwise. Denote by $N^\pi_K=\sum_{k=1}^K H^\pi_k$ the number of hits during the first $K$ requests for an object.

The following lemma holds, 

\begin{lemma}[Expected number of hits]
\label{lem1}
\begin{equation}
\label{lem:exp-hits}
\E\left[N^{HR-E}_K\right]\geq \E\left[N^\pi_K\right],\quad \forall \pi\in \Pi,\; \forall K\geq 1.
\end{equation}
\end{lemma}

\begin{proof} Fix $\pi\in \Pi$.
Given that a request for an object is made at time $t$ and given that the history $\mathcal{H}_t$ is known, this request is for object $i$ with the probability
\begin{equation}
\label{pi-t}
p_i(t)=\frac{\lambda^*_i(t)}{\sum_{j=1}^n \lambda^*_j(t)},
\end{equation}
Proof of (\ref{pi-t}) is given in the Appendix \ref{sub-app-pi}. This result relies on the conditional independence of the point processes $\{T_{i,k},k=1,2,\ldots\}$, $i=1,\ldots,n$. Observe that $p_i(t)$ does not depend on the caching policy in use.

By the definition of the HR-E policy 
\[
\sum_{i\in B^{HR-E}_k}\lambda^*_i(T_k) \geq \sum_{i\in B^{\pi}_k} \lambda^*_i(T_k),
\]
for $k=1,2,\ldots$. Therefore, for $k\geq 1$, 
\begin{eqnarray}
\E\left[H^{HR-E}_k\,|\, \mathcal{H}_{T_k}\right] = \sum_{i\in B^{HR-E}_k} \frac{\lambda^*_i(T_k) }{ \sum_{j=1}^n \lambda^*_j(T_k)}\geq
\sum_{i\in B^{\pi}_k}\frac{\lambda^*_i(T_k) }{ \sum_{j=1}^n \lambda^*_j(T_k)} = \E\left[H^{\pi}_k\,|\, \mathcal{H}_{T_k}\right].\label{lem:inq1}
\end{eqnarray}
Taking expectation on both sides of (\ref{lem:inq1}) to remove the conditioning yields 
$\E\left[H^{HR-E}_k\right]\leq \E\left[H^{\pi}_k\right]$.
Summing both sides of the latter inequality for $k=1,\ldots,K$ gives (\ref{lem:exp-hits}) from the definition of $N^\pi_K$.
\end{proof}

It is worth noting that Lemma \ref{lem1} holds for any non-stationary request object processes. We now study a more specific request arrival process and derive an upper bound on the object hit probability.

\subsection{Upper Bound on the Hit Probability for Stationary and Ergodic Object Arrival Processes}
\label{sub:hrb}

Let $\{\cdots<T_{i,-1}<T_{i,0}\leq 0<T_{i,1} <\cdots\}$ be the successive time epochs when object $i$ is requested. Define $X_{i,k}=T_{i,k}- T_{i,k-1}$ and introduce
the  two-sided sequence $\{X_{i,k},k\in \Z\}$ of inter-request times to object $i$.  We assume that $\{X_{i,k},k\in \Z\}$  is a stationary and ergodic sequence, and that
$X_{i,k}$ has a finite mean given by $\E[X_{i,k}]=\frac{1}{\lambda_i}$. We further assume that  the point processes $\{T_{i,k},k\in \Z\}$, $i=1,\ldots,n$, are mutually independent.

Define the point process $\{\cdots<T_{-1}<T_0<0 \leq T_1<T_2<\cdots\}$ obtained as the superposition of the $n$ point processes $\{T_{i,k},k\in \Z\}$, $i=1,\ldots,n$,
Define  $X_k:=T_k-T_{k-1}$, so that  $\{X_k,k\in \Z\}$ is the sequence of inter-request times for the point process $\{T_k,k\in\Z\}$.

The stationarity, ergodicity, and independence assumptions placed on point processes $\{T_{i,k},k\in \Z\}$, $i=1,\ldots,n$, imply that the sequence
 $\{(X_k,R_k),k\in \Z\}$ is stationary and ergodic (see {\em e.g.} \cite[pp. 33-34]{BB2003}).  
%

The stationary hit probability of policy $\pi\in \Pi$ is defined as  (cf. (\ref{def-Hk}))
\begin{equation}
\label{def-h}
h^\pi=\lim_{K\to\infty}\frac{1}{K}\sum_{k=1}^{K}  H^\pi_k \quad \hbox{a.s.,} 
\end{equation}
whenever this limit exists.

For any $\pi\in \Pi$, there exists a measurable mapping $\varphi^\pi: (R\times R)^\infty\to \{0,1\}$ such that $H^\pi_k=\varphi^\pi((T_j,R_j), j\leq k-1)$, which shows that the sequence 
$\{H^\pi_k,k\in \Z\}$ is stationary and ergodic ({\em e.g.} see \cite[Thm p. 62]{Phillips92}).  The ergodic theorem then yields the stationary hit probability
(see {\em e.g.} \cite[Thm 1]{Kingman68})
\begin{equation}
\label{h-2}
h^\pi=\Prob(H^\pi_k),
\end{equation}  
under $\pi$.
%
We are now in position to state and prove the main result of the paper.

\begin{theorem}[Stationary hit probability]\hfill

\[
h^{HR-E} \ge \max\limits_{\pi \in \Pi} h^{\pi}.
\]
\end{theorem}
\begin{proof}
Taking the expectation on both sides of (\ref{def-h}), using the fact that $h^\pi$ is a constant from (\ref{h-2}), and then  invoking the dominated convergence theorem, gives
\begin{equation}
\label{def-h2}
h^\pi= \lim_{K\to\infty}\frac{1}{K}\sum_{k=1}^{K}\E[H^\pi_k]= \lim_{K\to\infty}\frac{\E[N^\pi_K]}{K},
\end{equation}
with $N^\pi_K=\sum_{k=1}^{K} H^\pi_k$ (see Section \ref{ssec:number-hits}) the number of hits in $[0,T_{K}]$ or, equivalently due to the stationary, the number of hits in $K$ consecutive requests.
Proof is concluded by using Lemma \ref{lem1}.
\end{proof}


\begin{remark}
The computation of the HR-E based upper bound does not require the simulation of any caching policy.  At each request for an object, one can evaluate the HR values for all objects. One can then treat the request $R_k$ as a hit if the hazard rate of $R_k$ is among the top $B$ hazard rates at time $T_k$.
\end{remark}

\section{Variable Size Objects}\label{sec:vars}

We now assume object $i$ has size $s_i \in \mathbb{R^+}$ for all $i \in \mathcal{D}$ and the capacity of the cache is $B$ bytes. 

\subsection{Number of byte hits and fractional knapsack problem}\label{ssec:nb}
The setting and assumptions are that of Section \ref{ssec:number-hits} but  fractional caching (FC) is now allowed. We denote by $\Pi_{FC}$ the set of all FC policies.

For $\pi\in \Pi_{FC}$, let $V^\pi_k$ denote the number of bytes served from the cache at the $k$th request for an object. Let $x_{i,k}$ denote the fraction of object $i$ in the cache at the time of the $k$-th request. Then $V^\pi_k = s_i x^\pi_{i,k}$ if the request is for object $i$.  Let $W^\pi_K=\sum_{k=1}^K V^\pi_k$ denote the total number of bytes served from the cache during the first $K$ requests for an object.  

Given a request for an object is made at time $t$ and that the history $\mathcal{H}_t$ is known, we have already observed (see (\ref{pi-t})) that
 this request is for object $i$ with the probability $\lambda^*_i(t)/\sum_{j=1}^n \lambda^*_j(t)$. Therefore,
\begin{equation}
\label{pi-t2}
\mathbbm{E}[V^\pi_k\,|\, \mathcal{H}_{T_k}]=\sum_{i=1}^n \mathbbm{E}[V^\pi_k\,|\,  \mathcal{H}_{T_k}, \text{object $i$ is requested} ]\times \frac{\lambda^*_i(T_k)}{\sum_{j=1}^n \lambda^*_j(T_k)}=
\frac{\sum_{i=1}^n s_i x^\pi_{i,k}\lambda^*_i(T_k)}{\sum_{j=1}^n \lambda^*_j(T_k)}.
\end{equation}

Our goal is to find $\pi\in \Pi_{FC}$ that maximizes $\mathbbm{E}[V^\pi_k\,|\, \mathcal{H}_{T_k}]$ subject to the capacity constraint on the size of the cache. This can be done by solving the optimization problem,
\begin{subequations}\label{eq:fkp}
\begin{align}
\max \quad&\sum\limits_{i=1}^n s_ix_i\lambda^*_i(t)\\
\text{subject to} \quad&\sum\limits_{i=1}^n s_ix_i \le B \\
& 0\le x_i\le1,\quad i=1,\cdots,n,\label{eq:fkp33}
\end{align}
\end{subequations}
which is nothing but the Fractional Knapsack Problem (FKP)  \cite[Chapter 5.1]{Goodrich02}.  It is well known that its solution depends on the respective
values of the ratios $s_i \lambda_i^*(t)/s_i= \lambda_i^*(t)$ for $i=1,\ldots,n$. More specifically, assume that 
\begin{align}
\lambda^*_{i_1}(t) \ge \lambda^*_{i_2}(t)\ge \cdots \ge \lambda^*_{i_n}(t), \label{eq:order}
\end{align}
Then, the solution of \eqref{eq:fkp} is given by $x_{i_j}=1$ for $1\leq j\leq a:= \max\left\{a:s_{i_1} + s_{i_2}+ \cdots + s_{i_{a}}\leq B\right\}$,  
$x_{i_{a+1}}=(B-s_{i_1} - s_{i_2}- \cdots - s_{i_{a}})/s_{i_{a+1}}$, and
$x_{i_j}=0$ for $j=a+2, \ldots, n$.

Call HR-VB the rule which at any time $t$  places {\em entirely} in the cache objects with the highest hazard rates until an object cannot fit; if object $k$ is the first
one that cannot entirely fit  in the cache and objects $i_1,\ldots, i_j$ are already in the cache, then $B-\sum_{l=1}^j s_l$ bytes of object $k$ are cached. All other objects are not
cached. Then, by (\ref{pi-t2}), for any policy $\pi\in \Pi_{FC}$,
\[
\mathbbm{E}[V^{HR-VB}_k\,|\, \mathcal{H}_{T_k}] \geq \frac{\sum_{i=1}^n s_i x^\pi_{i,k}\lambda^*_i(T_k)}{\sum_{j=1}^n \lambda^*_j(T_k)}= 
\mathbbm{E}[V^{\pi}_k\,|\, \mathcal{H}_{T_k}].
\]
Removing the conditioning on $ \mathcal{H}_{T_k}$ yields $\mathbbm{E}[V^{HR-VB}_k]\geq \mathbbm{E}[V^{\pi}_k]$.  Summing both sides of this inequality
for $k=1,\ldots,K$ gives
\[
\mathbbm{E}[W^{HR-VB}_K]\geq \mathbbm{E}[W^{\pi}_K].
\]
\subsection{Number of object hits and $0$-$1$ knapsack problem}\label{sub:01}
The setting and assumptions are still that of Section \ref{ssec:number-hits} but we now assume that objects are  indivisible (IC). In particular, every object hit counts the same ({\em i.e.}, a hit for a large 1GB object and hit for a small 10B object both count as a "hit").  Denote by $\Pi_{IC}$ the set of all IC policies. 
Recall the definition of $H^\pi_k$ ($1$ if hit at $T_k$ and $0$ otherwise) and $N^\pi_K$ (number of hits in first $K$ requests) under $\pi\in \Pi_{IC}$.

Fix $\pi\in  \Pi_{IC}$. We have  by using (\ref{pi-t})
\begin{align}
\mathbbm{E}[ H^\pi_k\,|\, \mathcal{H}_{T_k}]&=\sum_{i=1}^n \mathbbm{E}[ H^\pi_k\,|\, \mathcal{H}_{T_k}, \text{object $i$ is requested at $T_k$}]\times \frac{\lambda^*_i(T_k)}{\sum_{j=1}^n \lambda^*_j(T_k)}\nonumber\\
&=\frac{1}{\sum_{j=1}^n \lambda^*_j(T_k)} \sum_{i=1}^n {\bf 1}(i\in B^\pi_k)\lambda^*_i(T_k),
\label{pi-t3}
\end{align}
where we recall that $B^\pi_k$ is the set of objects in the cache just before $T_k$  under $\pi$.

Hence, $\mathbbm{E}[ H^\pi_k\,|\, \mathcal{H}_{T_k}]$ can be maximized by solving the following 0-1 knapsack problem (KP),
\begin{subequations}\label{eq:kp}
\begin{align}
\max \quad&\sum\limits_{i=1}^n x_i\lambda^*_i(t)\\
\text{subject to} \quad&\sum\limits_{i=1}^n s_ix_i \le B \\
& x_i \in \{0, 1\},\quad i=1,\cdots,n.\label{eq:kp33}
\end{align}
\end{subequations}
Solving KP is NP-hard. However, the solution to the corresponding relaxed problem where the constraints in (\ref{eq:kp33}) are replaced by $x_i\in [0,1]$ for $i=1,\ldots,n$, 
serves as an upper bound for $\sum_{i=1}^n x_i\lambda^*_i(t)$. The latter is achieved  if the ratios $\{\lambda^*_i(t)/s_{i}\}_i$ are arranged in decreasing order, say
$\lambda^*_{i_1}(t)/s_{i_1}\geq \cdots \geq\lambda^*_{i_n}(t)/s_{i_n}$ and $x_{i_j}=1$ for $1\leq j\leq a$ where $a$ is defined in Section \ref{ssec:nb}, $x_{i_{a+1}}=(B-s_{i_1} - s_{i_2}- \cdots - s_{i_{a}})/s_{i_{a+1}}$, and $x_{i_j}=0$ for $j>a+1$ \cite[Chapter 5.1]{Goodrich02}.  

Call HR-VC the rule which, at any time $t$, places in the cache the objects in decreasing order of the  ratios $\{\lambda^*_i(t)/s_{i}\}_i$ until an object does not fit in the cache; 
if object $i_{a+1}$ is the first one that cannot entirely fit  in the cache and objects $i_1,\ldots, i_a$ are already in the cache, then with probability $p_{i_{a+1}} = (B-s_{i_1} - s_{i_2}- \cdots - s_{i_{a}})/s_{i_{a+1}}$ object $i_{a+1}$ is cached.  All subsequent objects according to this decreasing ordering are not cached.
Note that HR-VC does not meet the cache size constraint  as there is not enough room in the cache to fit entirely
object $i_{a+1}$. However, as mentioned in Section \ref{sec:model}, our goal is to upper bound $\E[H^\pi_k]$ and $\E[N^\pi_k]$. Let $x^*=(x^*_1,\ldots,x^*_n)$ be the solution of (\ref{eq:kp}). We have
\begin{align}
\mathbbm{E}[ H^{HR-VC}_k\,|\, \mathcal{H}_{T_k}]&= \frac{1}{\sum_{j=1}^n \lambda^*_j(T_k)} \left[\sum_{j=1}^{a_k} \lambda^*_{i_j}(T_k) + p_{i_{a_k+1}} \lambda^*_{i_{a_k+1}}(T_k) \right]\quad
\nonumber\\
&\geq \frac{1}{\sum_{j=1}^n \lambda^*_j(T_k)} \sum\limits_{i=1}^n x^*_i\lambda^*_i(t) \nonumber\\
&\geq 
 \frac{1}{\sum_{j=1}^n \lambda^*_j(T_k)} \sum_{i=1}^n {\bf 1}(i\in B^\pi_k)\lambda^*_i(T_k)=\mathbbm{E}[ H^\pi_k\,|\, \mathcal{H}_{T_k}],
\end{align}
with $a_k$ the last job that can be entirely cached at time $T_k-$ according to the decreasing ordering of the ratios $\{\lambda^*_i(T_k)/s_{i}\}_i$.
Removing the conditioning on $ \mathcal{H}_{T_k}$ yields $\mathbbm{E}[ H^{HR-VC}_k]\geq \mathbbm{E}[ H^{\pi}_k]$.  Summing both sides of this inequality
for $k=1,\ldots,K$ gives
\begin{align}
\mathbbm{E}[ N^{HR-VC}_K]\geq \mathbbm{E}[ N^{\pi}_K].\label{eq:hr-var-sz-hit-prob}
\end{align}
\subsection{Upper Bound on Object Hit Probability for Stationary and Ergodic Object Arrival Processes}\label{sub:hrb-vs}
Using similar arguments as discussed in Section \ref{sub:hrb},  one can define the object hit probability under HR-VC (cf. Section \ref{sub:01}) as
\begin{equation}
\label{def-h2-vs1}
h^{HR-VC}=  \lim_{K\to\infty}\frac{1}{K}\sum_{k=0}^{K-1}\mathbbm{E}[ H^{HR-VC}_k] = \lim_{K\to\infty}\frac{\mathbbm{E}[ N^{HR-VC}_K]}{K},
\end{equation}
Now combining  \eqref{eq:hr-var-sz-hit-prob} and \eqref{def-h2-vs1} we obtain
\begin{equation}
\label{def-h2-vs2}
h^\pi= \lim_{K\to\infty}\frac{\E[N^\pi_K]}{K} \le \lim_{K\to\infty}\frac{\mathbbm{E}[ N^{HR-VC}_K]}{K} = h^{HR-VC}.
\end{equation}

\section{Specific Request Arrival Processes}\label{sec:disc}
Below we consider four specific object request processes each with equal size, some of which explicitly account for the temporal locality in requests for objects.  Note that unlike in Sections \ref{ssec:Poisson}, \ref{subsec:on-off}, \ref{subsec:snm} (Poisson, on-off, shot noise) requests to {\em different} objects are correlated in Section \ref{subsec:mmpp} (MMPP).

\subsection{Poisson Process}
\label{ssec:Poisson}
We consider the case when successive requests to object $i$ ($i=1,\ldots,n$) occur according to a Poisson process with rate $\lambda_i>0$ and these $n$ Poisson processes are mutually independent. This is the standard Independence Reference Model (see Section \ref{ssec:IRM})  where references to all objects are independent rvs.
Without loss of generality assume that $\lambda_1\geq \cdots \geq \lambda_n$.

Under HR , at any time only objects $1,\ldots,B$ are in the cache. Therefore, the hit probability $h^{HR}_i$ for object $i$ is $h^{HR}_i={\bf 1}(i\leq B)$
and the hit rate $r^{HR}_i$ for object $i$ is $r^{HR}_i=\lambda_i {\bf 1}(i\leq B)$. The overall hit probability $h^{HR}$  and hit rate $r^{HR}$ are given by
\[
h^{HR}=\frac{1}{\Lambda}\sum_{i=1}^B\lambda_i ,\quad r^{HR}= \sum_{i=1}^B \lambda_i, 
\]
where $\Lambda:=\sum_{i=1}^n \lambda_i$.
\subsection{On-Off Request Process}\label{subsec:on-off}
\label{ssec:onoff}
\begin{figure}[!htbp]
\centering
\includegraphics[width=0.4\linewidth]{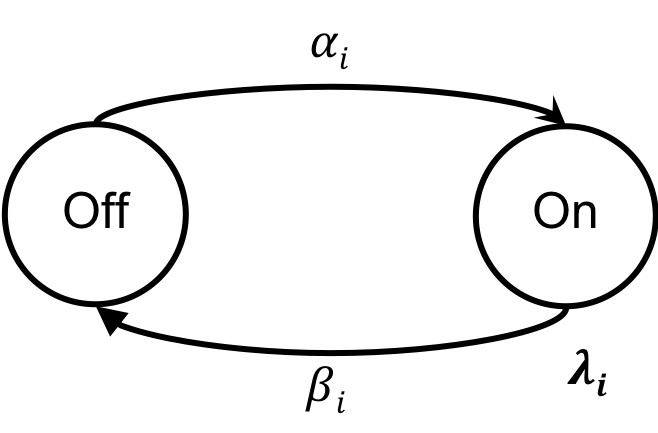}
\vspace{-0.1in}
\caption{On-Off Request Process}
\label{onoff}
\vspace{-0.1in}
\end{figure} 
The object popularity dynamics in caching systems can be effectively captured by using a stationary, on-off traffic model \cite{Garetto2015}. 
More specifically, we assume that successive requests to object $i$ occur according to a Poisson process with rate $\lambda_i>0$ 
when the underlying on-off process depicted in Figure \ref{onoff} is in state $1$ ($X_i=1$) and that no request occurs when this process is in state $0$ ($X_i=0$). 
The stationary distribution of this on-off process is given by $\boldsymbol \pi_i := [\pi_{i,0},\pi_{i,1}] = [\beta_i/(\alpha_i+\beta_i),\alpha_i/(\alpha_i+\beta_i)]$.
We assume that these $n$ on-off processes are mutually independent.
Without loss of generality, assume that $\lambda_1\ge\lambda_2\ge\ldots\ge \lambda_n$ and define $\Lambda=\sum_{i=1}^n \lambda_i$. Below, we derive expressions for the hit rate and hit probability under HR. 

Due to the way HR behaves, we may assume without loss of generality that object $i$ is never in the cache when $X_i=0$. 
With this, at any time at most  $B$ of the most popular objects are in the cache among all objects whose associated on-off process is in state $1$. Therefore, if $i>B$ the hit probability
$h_i^{HR}$ for object $i$ is given by
\begin{align*}
h^{HR}_i &= \Prob(\text{at most  $B-1$ on-off processes are in state $1$ among on-off processes $1,\ldots,i-1$})\nonumber\\
&= \sum_{k=0}^{B-1}\Prob(\text{exactly $k$ on-off processes are in state $1$ among on-off processes $1,\ldots,i-1$})\nonumber\\
&= \sum_{k=0}^{B-1}\sum_{i_1,\ldots,i_k \in \{1,2,\ldots,i-1\}\atop i_1<i_2<\cdots <i_k} \prod_{l=1}^k \pi_{i_l,1}\prod_{m\in \{1,\ldots,i-1\}-(i_1,\ldots,i_k)} \pi_{i_m,0},
\label{hHR}
\end{align*}
and $h^{HR}_i=1$ if $i\leq B$. The hit rate $r^{HR}_i$ for object $i$ is $r^{HR}_i=\lambda_i \pi_{i,1} h^{HR}_i$.

The overall hit probability $h^{HR}$ and the overall hit rate $r^{HR}$ are given by 
\begin{equation}
h^{HR}=\sum_{i=1}^n  \frac{\lambda_i \pi_{i,1}}{\sum_{j=1}^n \lambda_j \pi_{j,1}}h^{HR}_i\quad \hbox{and}\quad
r^{HR}=h^{HR}\sum_{i=1}^n \lambda_i \pi_{i,1}.
\label{hit-onoff}
\end{equation}
Assume that the ratios $\alpha_i/(\alpha_i+\beta_i)$ and $\beta_i/(\alpha_i+\beta_i)$ do not depend on $i$ and define $\rho=\alpha_i/(\alpha_i+\beta_i)$ for all $i$. 
This occurs, for instance, if all $n$ on-off processes have the same transition rates with $\alpha_i=\alpha$ and $\beta_i=\beta$ or if $\alpha_i=\alpha \theta_i$ and 
$\beta_i =\beta \theta_i$ for all $i$. Then, for $i>B $,
\begin{align*}
h^{HR}_i &= \sum_{k=0}^{B-1}\sum_{i_1,\ldots,i_k \in \{1,2,\ldots,i-1\}\atop i_1<i_2<\cdots <i_k} \rho^k (1-\rho)^{i-1-k} 
=(1-\rho)^{i-1} \sum_{k=0}^{B-1} \left(\frac{\rho}{1-\rho}\right)^k {i-1\choose k},
\end{align*}
so that
\[
h^{HR} = \frac{\sum_{i=1}^B \lambda_i}{\Lambda}+\frac{1}{\Lambda}\sum_{i=B+1}^n\lambda_i (1-\rho)^{i-1} \sum_{k=0}^{B-1} \left(\frac{\rho}{1-\rho}\right)^k {i-1\choose k},
\]
and $r^{HR}=\rho\sum_{i=1}^B \lambda_i +\rho \sum_{i=B+1}^n\lambda_i (1-\rho)^{i-1} \sum_{k=0}^{B-1} \left(\frac{\rho}{1-\rho}\right)^k {i-1\choose k}$. 

We now propose a recursive approach for computing the hit probability and the hit rate with a much lower computational complexity than the general formulas in
(\ref{hit-onoff}). 

The recursions are based on available objects in the catalog, starting from the situation where only object $1$ is available,  moving to the situation
where objects $1$ and $2$ are available, etc. up  to the final situation where all $n$ objects are available. Introduce the following variables,
\begin{align*}
p_{l,k}&=\Prob[\text{cache occupancy is }k  \,|\, \text{catalog is composed of $l$ more popular objects}],\nonumber\\
r_{l,k}&= \text{Hit rate when cache occupancy is $k$ given catalog is composed of $l$ more popular objects }.
\end{align*}
When $l=1$ then $p_{1,0}=\pi_{1,0}$, $p_{1,1}=\pi_{1,1}$, $r_{1,0}=0$, and $r_{1,1}=\lambda_1$ from our convention that object $1$ is not in the cache when $X_1=0$.
It is easy to verify that under HR the following recursions hold true for the occupancy probabilities,
\begin{align}
p_{l,0} &= p_{l-1,0}\pi_{l,0},\quad l = 2,\cdots, n,\nonumber\\
p_{k,k} &= p_{k-1,k-1}\pi_{k,1},\quad k=1,\ldots, B ,\nonumber\\
p_{l,k} &= p_{l-1,k-1}\pi_{l,1} + p_{l-1,k,}\pi_{l,0},\quad 0 < k < \min(l,B), l = 1,\cdots, n,\nonumber\\
p_{l,B} &= p_{l-1,B-1}\pi_{l,1} + p_{l-1,B},\quad l=B+1,\cdots,n.
\end{align}
Similarly, the  following recursions hold true for the hit rates,
\begin{align}
r_{l,0} &= 0,\quad l = 2,\cdots, n,\nonumber\\
r_{k,k} &= r_{k-1,k-1}+ \lambda_{k},\quad  k=2,\ldots B,\nonumber\\
r_{l,k} &=\frac{p_{l-1,k-1}\pi_{l,1}(r_{l-1,k-1}+\lambda_{l}) + p_{l-1,k}\pi_{l,0}r_{l-1,k}}{p_{l-1,k-1}\pi_{l,1} + p_{l-1,k}\pi_{l,0}},\quad 0 < k < \min(l,B), l = 1,\cdots, n,\nonumber\\
r_{l,B} &= \frac{p_{l-1,B-1}\pi_{l,1}(r_{l-1,B-1}+\lambda_{l}) + p_{l-1,B}r_{l-1,B}}{p_{l-1,B-1}\pi_{l,1} + p_{l-1,B}},\quad l=B+1,\cdots,n.
\end{align}
Once the above recursions have been solved, the overall hit rate $r^{H}$ and hit probability $h^{HR}$ under HR are given by
\begin{align}
r^{HR} = \sum\limits_{k=1}^B p_{n,k}r_{n,k}\quad \hbox{and}\quad h^{HR} = \frac{r^{HR}}{\sum\limits_{l=1}^n\lambda_l\pi_{l,1}}.
\end{align}
%
%
%
\subsection{Markov Modulated Poisson Process}\label{subsec:mmpp}

We assume that the environment is modulated by a stochastic process  $\{X(t), t\geq 0\}$ taking values in a denumerable set ${\mathcal E}$.  
Let $0<t_1<t_2<\cdots$ be the successive jump times of the process  $\{X(t)\}_t$. We assume that $X(0)$ is known. 

Let $0<T_{i,1}<T_{i,2}<\cdots< T_{i,k}<\cdots$ be the sucessive times when object $i$ is requested for $k=1,2,\ldots$. 
We assume that (by convention $T_{i,0}=0$ for all $i$)
\begin{align}
&\Prob\left( T_{1,k}-T_{1,k-1}> y_1,\ldots,   T_{n,k}-T_{n,k-1}>y_n\,|\, T_{1,k-1}, \ldots, T_{n,k-1}, X(s), 0\leq s\leq \max_{1\leq i\leq n} (T_{i,k-1}+y_i)\right)\nonumber\\
&= \prod_{i=1}^n \exp\left(-\int_{T_{i,k-1}}^{T_{i,k-1}+y_i} \lambda_i({X(s))} ds\right), \label{def-request-processes}
\end{align}
with $\lambda_i(x)>0$ for all $i=1,\ldots,n$ and $x\in {\mathcal E}$.  

In words, in $[t_m, t_{m+1})$ the object request processes $\{T_{1,k}\}_k, \ldots, \{T_{n,k}\}_k$ 
are mutually independent Poisson processes with intensities $\lambda_1(X(t_m)), \ldots, \lambda_n(X(t_m))$, respectively.\\

We assume that the cache policy knows the state of the environment at any time. 
Given that a request is made at time $t$, the object requested  at time $t$ is object $i$ with the probability
\begin{equation}
p_i(t)= \frac{\lambda_i(x)}{\sum_{j=1}^n \lambda_j(x)},
\label{pit}
\end{equation}
if $X(t)=x$.

We denote by ${\rm HR}$ the cache policy which at any time $t$ caches the $B$ objects with the largest $\lambda_i(X(t))$.   
Let $B_k^\pi$ denote the state of the cache just before the $k$th object is requested at time $T_k$ under policy $\pi$. 

Under $\pi$ there is a hit at time $T_k$  if the requested object is in the cache ($H^\pi_k=1$) and a miss if not ($H^\pi_k=0$).  Hence, under $\pi$ the number of hits, $N^\pi_K$, after $K$ requests is $N_K^\pi=\sum_{k=1}^K H^\pi_k$.


\begin{lemma} \label{lemma:mmpp}
For any admissible caching policy $\pi$,
\[
\E[N^{\rm HR}_K]\geq \E[N^\pi_K], \quad \forall K\geq 1.
\]
\end{lemma}
{\bf Proof.}
Proof is given in Appendix \ref{sub-app-mmpp-lemma}.  
\hfill\done\\

Assume that $\{X(t)\}_t$  has a stationary distribution independent of its initial state, denoted by  $\{\gamma(x), x\in {\mathcal E}\}$. Let $h^\pi$ denote the stationary hit probability under $\pi$.  We have
\begin{eqnarray}
h^\pi&=&\lim_{k\to\infty} \E[H^\pi_k] = \lim_{k\to\infty}\sum_{x\in {\mathcal E}} \int_{t=0}^\infty dP(T_k<t | X(T_k-)=x) \Prob(X(T_k-)=x) \frac{\sum_{i\in B_k^\pi}\lambda_i(x)}{ \sum_{j=1}^n \lambda_j(x)}\nonumber\\
&=& \lim_{k\to\infty}\sum_{x\in {\mathcal E}} \Prob(X(T_k-)=x) \frac{\sum_{i\in B_k^\pi}\lambda_i(x)}{ \sum_{j=1}^n \lambda_j(x)}\nonumber\\
&=&\sum_{x\in {\mathcal E}} \gamma(x)  \frac{\sum_{i\in B^\pi}\lambda_i(x)}{ \sum_{j=1}^n \lambda_j(x)},
\label{id-hpi}
\end{eqnarray}
with $B^\pi$ the stationary version of $B^\pi_k$.
\begin{theorem}
\[
h^{\rm HR}\geq \max_{\pi} h^\pi.
\]
\end{theorem}
{\bf Proof.}  Follows from (\ref{id-hpi}) and $\sum_{i\in B^{\rm HR}}\lambda_i(x) \geq \sum_{i\in B^\pi}\lambda_i(x)$ for all $x\in {\mathcal E}$. \hfill\done

We also obtain the following analytic expression for hit probability for hazard rate based upper bound,
\begin{align}
h^{\rm HR} = \sum_{x\in {\mathcal E}} \gamma(x)  \frac{\sum_{i=1}^{B}\lambda_{[i]}(x)}{ \sum_{j=1}^n \lambda_j(x)},\label{eq:closedmmpp}
\end{align}
\noindent where $\lambda_{[1]}(x) \ge \lambda_{[2]}(x) \ge \cdots \ge \lambda_{[n]}(x)$ for all $x\in {\mathcal E}.$

\subsection{Shot Noise Model}
\label{subsec:snm}
Another traffic model, named Shot Noise Model (SNM) \cite{Traverso2015}, has been proposed to capture the temporal locality observed in real traffic in  caching systems {\em e.g.} in Video on Demand (VoD) systems. The primary idea of the SNM is to represent the overall request process as the superposition of many independent time inhomogeneous Poisson processes or shots, each referring to an individual object. In particular, the request process for object $i$ is described by an inhomogeneous Poisson process of instantaneous rate
\begin{align}
\lambda_i^{inst}(t) = V_i\lambda_i(t-\tau_i),\label{inhomo-poisson}
\end{align}
\noindent where $\tau_i$ is the time instant at which object $i$ is first requested, $V_i$ denotes the expected number of requests generated by object $i$ and $\lambda_i(x)$ is the popularity profile of object $i$ over time. 
It is easy to check that the instantaneous hazard rate associated with object $i$ can be calculated as
\begin{align}
\lambda_i^{*}(t) = \lambda_i^{inst}(t),\label{inhomo-poisson2}
\end{align}
\noindent Thus the results from Sections \ref{ssec:number-hits} and \ref{sub:01} directly apply with $\lambda_i^{*}(t) = \lambda_i^{inst}(t).$

%
%
%

\section{Numerical Results}\label{sec:perf}
Via simulations we compare the stationary object hit probabilities of various online policies (Section \ref{ssec:policies}) to that of our proposed upper bound   -- referred to as {\em HR (based) upper bound} -- to B\'el\'ady's upper bound (BELADY) and to a third bound (FOO -- see Section \ref{ssec:ub}). This is done  for a number of arrival processes of object requests (Section \ref{ssec:arrival}), for equal and different sized objects
(Section \ref{ssec:size}), and for several cache sizes. We first present the experimental setup and then discuss the results.

\subsection{Experimental setup}
\label{ssec:setup}

\subsubsection{Investigated online policies}
\label{ssec:policies}
Several caching policies have been used to generate Figures \ref{plots}-\ref{real-trace}.
The well-known LRU, FIFO, and RANDOM  cache replacement policies discard the least recently used items first,  evicts objects in the order they were added, and randomly selects an object and discards it to make space when necessary, respectively. The STATIC policy keeps forever in the cache the $B$ objects which have the largest arrival rates. 
Notice that the HR based bound and the hit probability under STATIC are equal when successive requests for each object follow a Poisson process (Section \ref{ssec:Poisson}). We also consider the Greedy-Dual-Size-Frequency (GDSF) policy \cite{GDSF98} which combines recency with frequency and size to improve upon LRU. Last,  the AdaptSize policy \cite{Berger17} uses an adaptive size threshold with admission control preferring  admission of small sized objects.
\subsubsection{Upper bounds on object hit probability}
\label{ssec:ub}
Aside our HR based upper bound which applies to both equal and variable sized objects, two other upper bounds on the object hit probability proposed in literature are used, 
B\'el\'ady's offline upper bound (BELADY, Section \ref{ssec:offline-ub}) for equal sized objects and a flow based offline optimal  (FOO) \cite{berger18} for different sized objects. FOO upper bound is computed by representing caching as a min-cost flow problem.
\subsubsection{Arrival process of object requests in Figures \ref{plots}-\ref{real-trace}}
\label{ssec:arrival}
In each display in Figures \ref{plots}-\ref{plots-var}, request processes for objects $i=1,\ldots,n$ are independent renewal processes with  IRT distributions shown in Table \ref{tbltraff}. 
More specifically, 
in Figure \ref{plots}(a) (resp. Figures \ref{plots}(b)-\ref{plots}(f)) the request process for object $i=1,\ldots,n$ has an exponential IRT (resp. Generalized Pareto, Uniform, Hyperexponential, Gamma, Erlang) with arrival rate $\lambda_i$ drawn from a Zipf distribution with parameter $0.8$ (see  last column of Table \ref{tbltraff}); similarly, in Figure \ref{plots-var}(a) 
(resp. Figures \ref{plots-var}(b)-(c)) the IRT has an exponential (resp. Generalized Pareto, Uniform) distribution with arrival rate $\lambda_i$ drawn from a Zipf distribution with parameter $0.8$.

\begin{table*}[h]
\resizebox{\columnwidth}{!}{
\begin{tabular}{|c| c| c| c| }
\hline
  Inter-request time &Hazard Rate& $\Prob(IRT_i<t)$ &Arr. rate $\lambda_i$ ($=1/\E[IRT_i]$)\\
distribution  (IRT)&& &drawn from Zipf\,(0.8)\\
\hline
Exponential&CHR&$1-e^{-\lambda_it}$&$\lambda_i$\\
\hline
Generalized Pareto&  DHR&    $1 -(1+\frac{k_it}{\sigma_i})^{-\frac{1}{k_i}}$, $k_i=0.48$ & $\frac{1-k_i}{\sigma_i}$\\
\hline
Hyperexponential$^\star$& DHR&   $1-\sum\limits_{j=1}^{2}p_{ji}e^{-\theta_{j,i}t}$  & \\
& & $p_{1,i}+p_{2,i}=1$ & $\frac{1}{2\nu_i}$\\
&  &  $p_{1,i}/\theta_{1,i}=p_{2,i}/\theta_{2,i}:=\nu_i$ &\\ 
&  &  $SCV_i=\text{var}(IRT_i)/\E[IRT_i]^2=2$  &\\
\hline
Uniform&IHR&   $\frac{t}{b_i}$ &   $\frac{2}{b_i}$\\
\hline
Gamma &  DHR ($k_i<1$)&   $\frac{1}{\Gamma(k_i)}\gamma(k_i,\frac{t}{\theta_i})$, $k_i=0.5$ &$\frac{2}{\theta_i}$\\
\hline
Erlang  &  IHR&  $\frac{\gamma(k_i,\mu_it)}{(k_i-1)!}$, $k_i=0.2$ &$\frac{\mu_i}{2}$ \\
\hline
\end{tabular}
}
\caption{Inter-request time (IRT) distributions of the renewal request arrival  processes in Figures \ref{plots}-\ref{plots-var} and their properties
(CHR = Constant hazard rate, IHR = Increasing hazard rate, DHR = Decreasing hazard rate). 
$^\star$Once arrival rate $\frac{1}{2\nu_i}$ is known one finds $p_{1,i}=(1- \sqrt{(SCV_i-1)/(SCV_i+1)})/2$ under the constraints.}
\label{tbltraff}
\vspace{-0.2in}
\end{table*}
In Figure \ref{on-off-trace}(a)-(b) the  arrival request process for object $i$ ($i=1,\ldots,n$) is generated via an on-off process (see Section \ref{ssec:onoff}) 
and these $n$ on-off processes are mutually independent. The transition rates for on-off process $i$ are $\alpha_{i} = 1/T_{OFF}$ and $\beta_{i} = 1/T_{ON}$,
with $T_{ON} = 7$ (days) and $T_{OFF} = 9T_{ON}$. The arrival rate $\lambda_i$ in the on-state is given by $\lambda_i = V/T_{ON}$, where
$V$ is drawn from a Pareto distribution with pdf $f_{V} (v) = \beta V_{min}^\beta/v^{1+\beta}$, $\mathbb{E}[V] = 10$, and $\beta = 2$  \cite{Garetto2015}.\\

In Figure \ref{mmpp} requests for objects are generated according to a two-state MMPP (see Section \ref{subsec:mmpp}).
Without loss of generality (W.l.o.g.), call $1$ and $2$ these two states. Let $\alpha$ and $\beta$ be the state transition rate from state $1$ to $2$ and from state $2$ to $1$, respectively. The stationary state probabilities are  $\gamma(1) =\beta/(\alpha+\beta)$ and  $\gamma(2) =\alpha/(\alpha+\beta).$  
In the simulations, we took
$\alpha=2\times 10^{-3}$ and $\beta=1.6\times 10^{-3}$.  In state $j$, successive requests for object $i$ are generated according to a Poisson process with rate
$\lambda_i(j)$ for $j=1,2$. 
In state $1$, we assume that object arrival rates $\lambda_i(1),\ldots, \lambda_n(1)$ each follows a Zipf distribution with parameter $0.8$;
W.l.o.g assume that $\lambda_{1}(1) > \lambda_{2}(1) >\cdots > \lambda_{n}(1)$. 
In state $2$,  we assume that object arrival rates are given by  $\lambda_i(2) = \lambda_{n+1-i}(1)$ for $i=1,\cdots,n$.\\
 
In Figure \ref{shot-noise-trace} requests for objects are generated by $n$ independent shot noise processes (see Section \ref{subsec:snm}). Objects belong to four different classes. Objects in class $c$ ($c=1,\ldots,4$) become available in the system  at times $\tau_1(c)<\cdots<\tau_{n_c}(c)$ of a 
homogeneous Poisson process  with rate $\gamma_c = \mathbb{E}[V_c]/\mathbb{E}[L_c]$. The SNM associated with the $i$th object of class $c$, which becomes available at time $\tau_i(c)$, 
has intensity  $\lambda_i(t)=(V_i/\alpha_c)e^{-(t-\tau_i(c))/\alpha_c}$,  with $\alpha_c = \frac{0.5}{0.8}\mathbb{E}[L_c]$ and where $V_i$ is chosen according to a Poisson distribution with rate $\mathbb{E}[V_c]$. Values of  $\E[V_c]$ (expected number of requests for a class $c$ object) and $\E[L_c]$ (expected lifespan of a class $c$ object)
are given in Table \ref{table:sn}. This model has been obtained by the authors of \cite{Traverso2013}
from their so-called Trace 1, which contains $n=\sum_{c=1}^4 n_c=143871$ objects (cf. $4$th column of  Table \ref{table:sn}).
\begin{table}[h]
	\begin{center}
		\begin{tabular}{ |c|c|c| c| c|} 
			\hline
			Class id (c) & $\mathbb{E}[L_c]$ &  $\mathbb{E}[V_c]$ & Catalog size ($n_c$)\\ 
			\hline
			Class 1 & $1.14$ & $86.4$ & $29481$\\
			Class 2 & $3.36$ & $41.9$ &  $45570$\\
			Class 3 & $6.40$ & $59.5$ &  $27435$\\
			Class 4 & $10.53$ & $36.9$ &  $41385$\\
			\hline
		\end{tabular}
		\caption{Parameters of the shot-noise models in Figure \ref{shot-noise-trace}.}
		\label{table:sn}
	\end{center}
\end{table}

In Figure \ref{real-trace} we use requests from a Web access trace collected from a gateway router at IBM research lab \cite{zerfos13}. We filter the trace such that each object has been requested at least a hundred times. The filtered trace contains $3.5\times10^6$ requests with an object catalog of size $n = 5638$.  
Various parametric and non-parametric estimators have been developed in the literature to estimate the hazard rate \cite{Wang05, Singpurwalla1983}. Here, we adopt a parametric estimator model and assume that the inter-request times for each object are independent and identically distributed non-negative random variables. Note that the Web and storage traffic inter-request times and access patterns are well modeled by  heavy-tailed distributions \cite{DOWNEY2005790,Gracia-Tinedo2015}. Hence, we fit the density of inter-request times of each object to a Generalized-Pareto distribution using the maximum likelihood estimation technique and estimate the hazard rate for each object accordingly. 

\subsubsection{Size of objects}
\label{ssec:size}
Both objects of equal size and variable size are considered.  In the former the size of each object is equal to $1$ and in the latter the size of each object is drawn independently according to a bounded Pareto distribution with Pareto shape parameter $1.8$, minimum object size of $5$Mb and maximum object size of $15$Mb.
When all objects have  same size the size of the cache is  expressed in number of objects and it is expressed in Mb when objects have different sizes.
\subsection{Discussion}
\label{ssec:discussion}
\begin{figure*}[]
\centering
\hspace{-0.3cm}
\begin{minipage}{0.32\textwidth}
\includegraphics[width=1\textwidth]{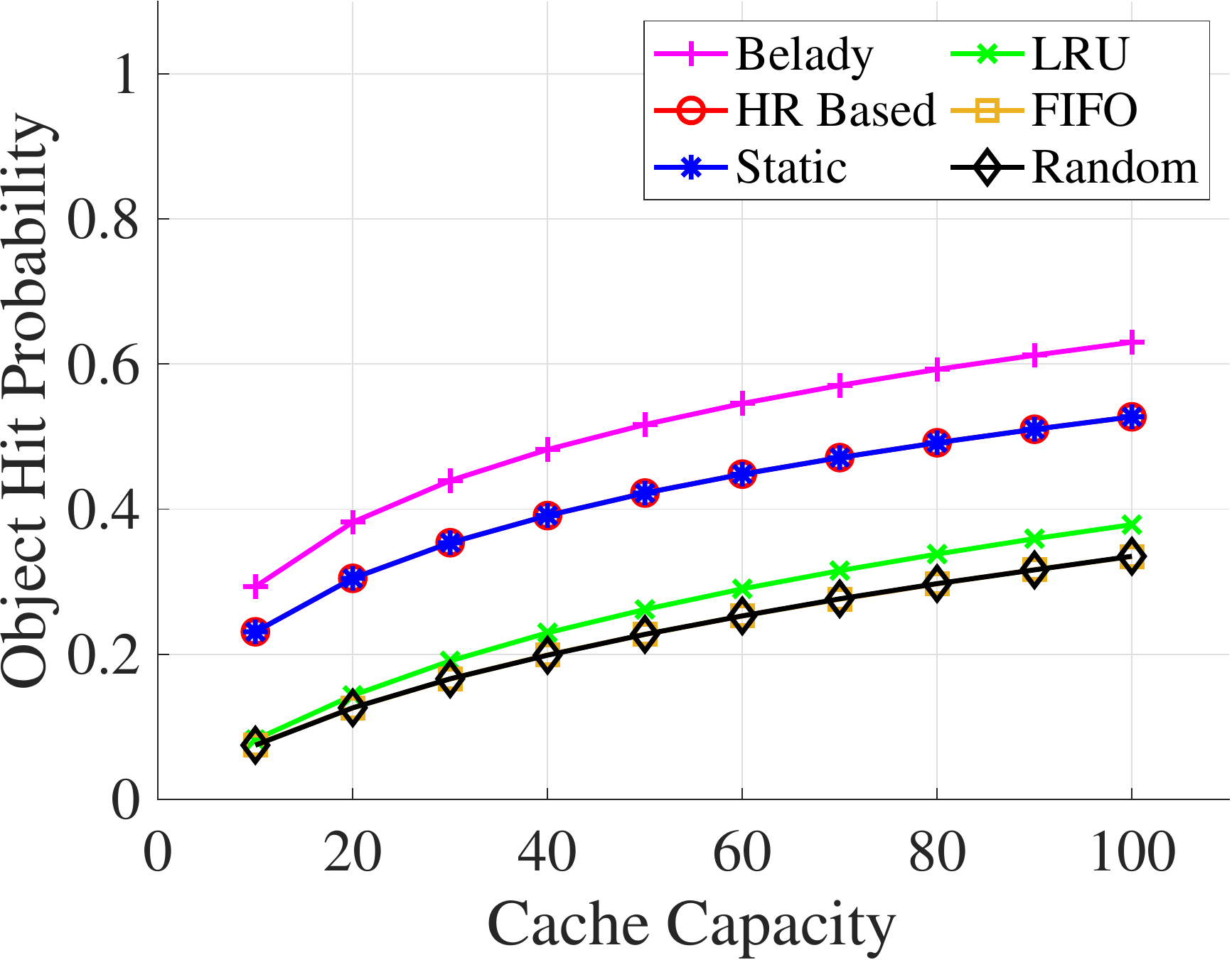}
\subcaption{Exponential (CHR)}
\end{minipage}
\begin{minipage}{0.32\textwidth}
\includegraphics[width=1\textwidth]{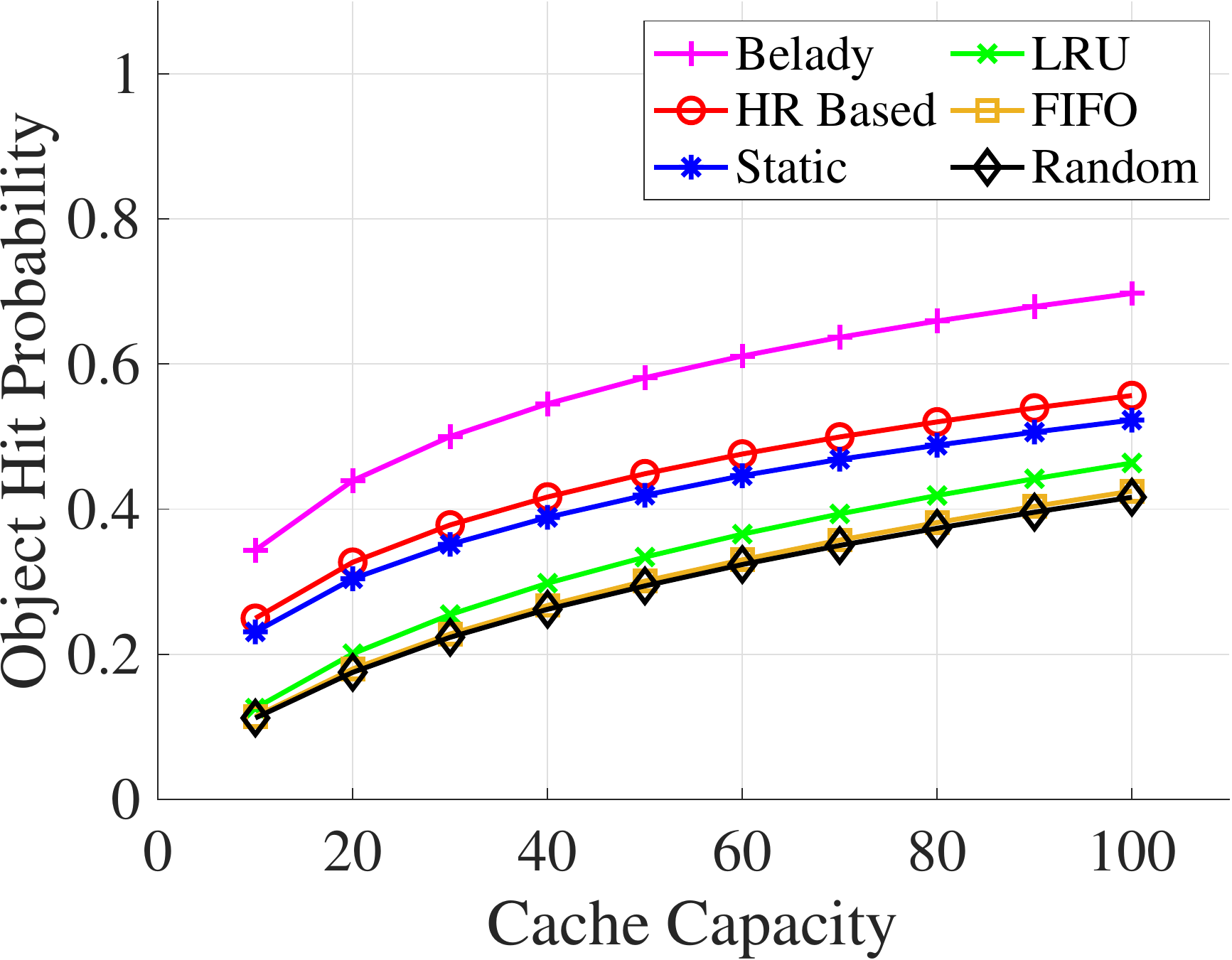}
\subcaption{Generalized Pareto (DHR)}
\end{minipage}
\begin{minipage}{0.32\textwidth}
\includegraphics[width=1\textwidth]{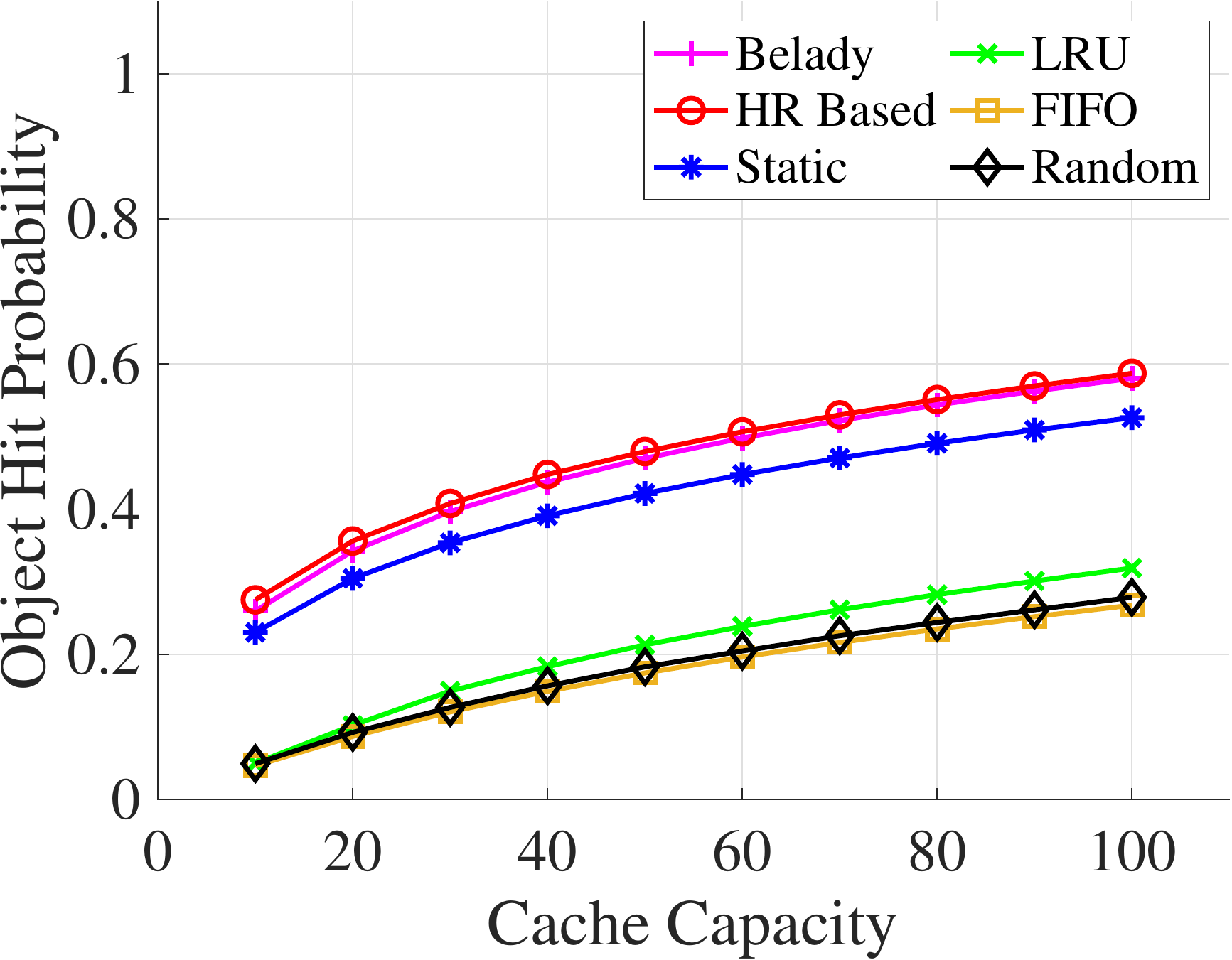}
\subcaption{Uniform (IHR)}
\label{pltuni}
\end{minipage}
\begin{minipage}{0.32\textwidth}
\includegraphics[width=1\textwidth]{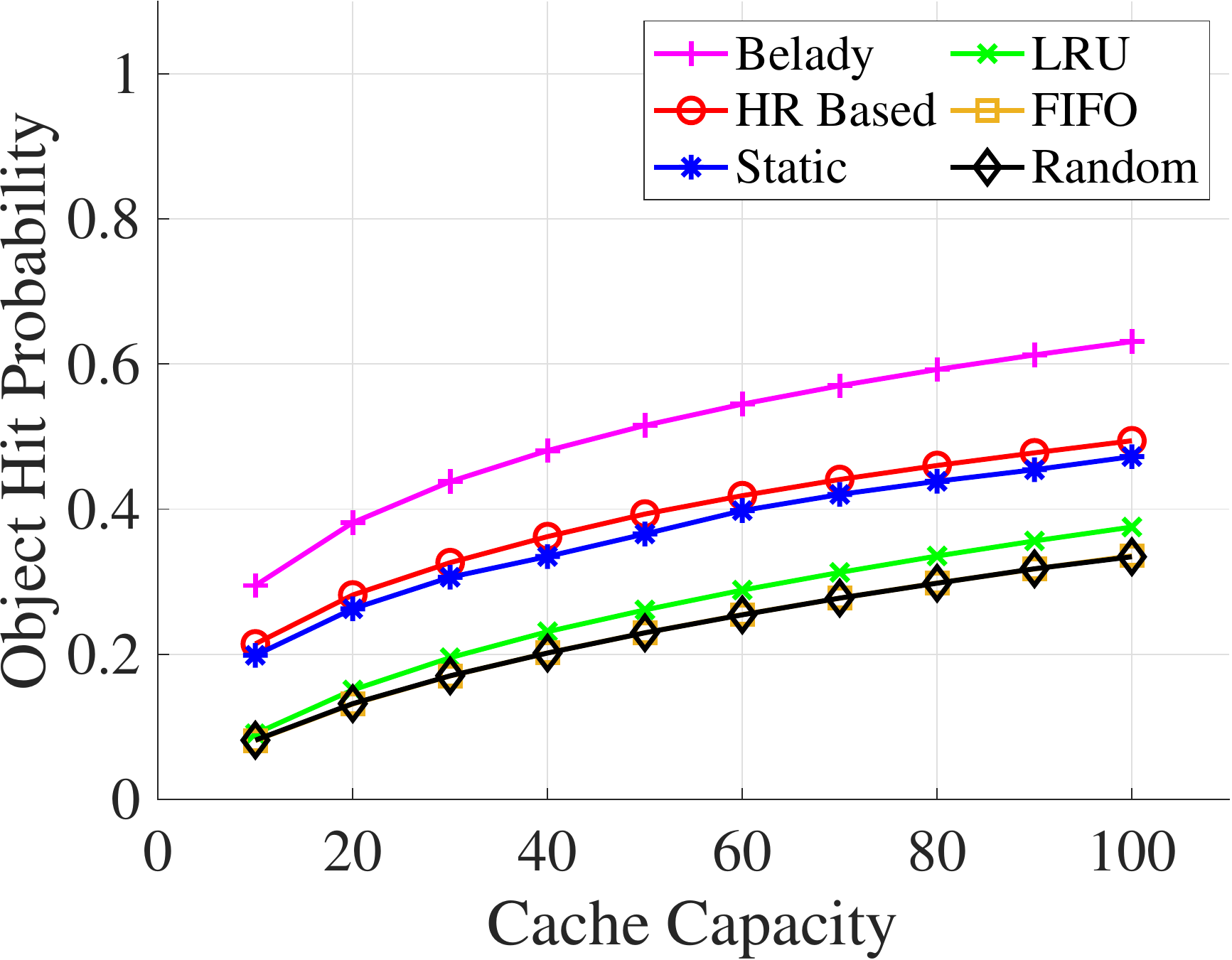}
\subcaption{Hyperexponential (DHR)}
\end{minipage}
\begin{minipage}{0.32\textwidth}
\includegraphics[width=1\textwidth]{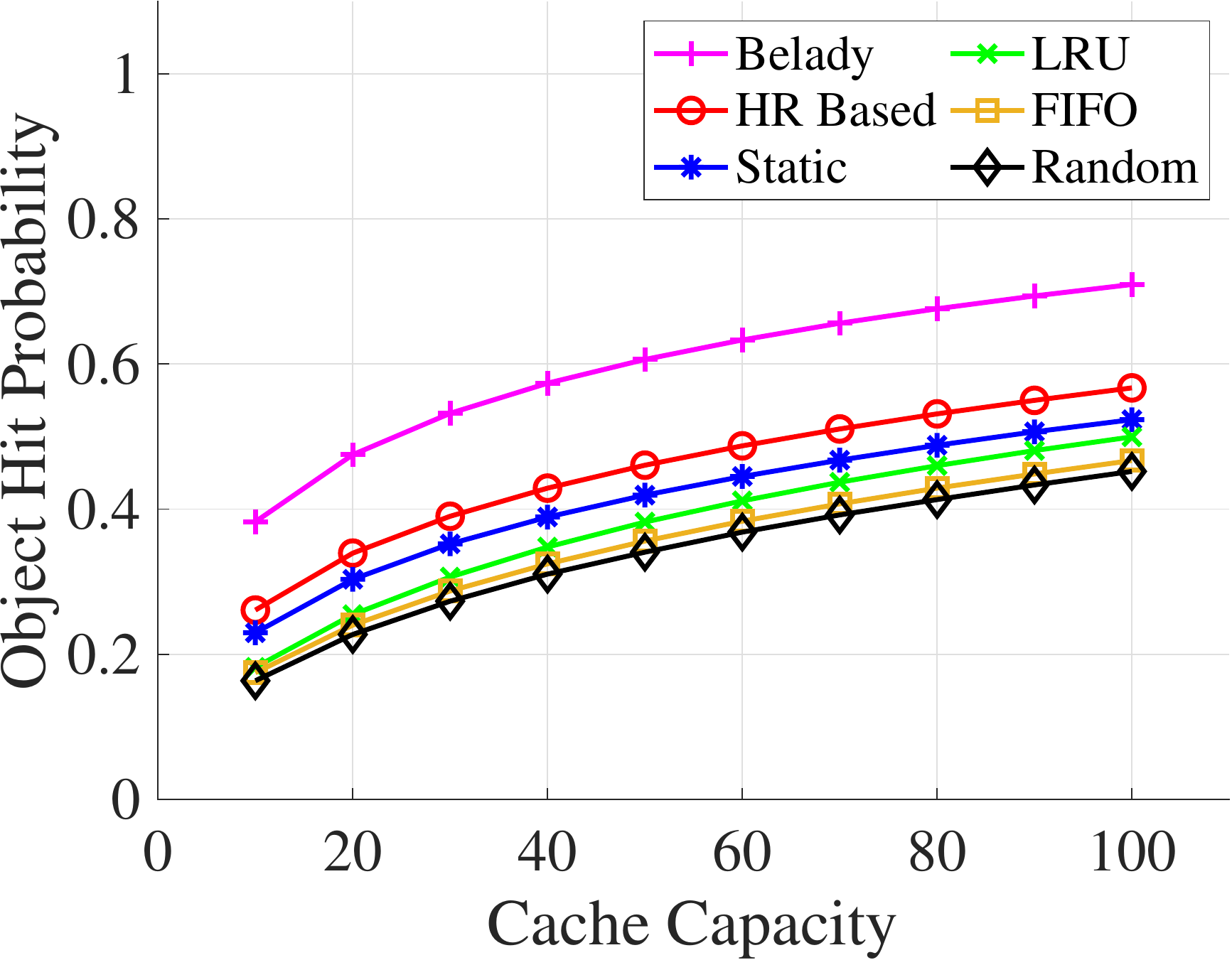}
\subcaption{Gamma (DHR)}
\end{minipage}
\begin{minipage}{0.32\textwidth}
\includegraphics[width=1\textwidth]{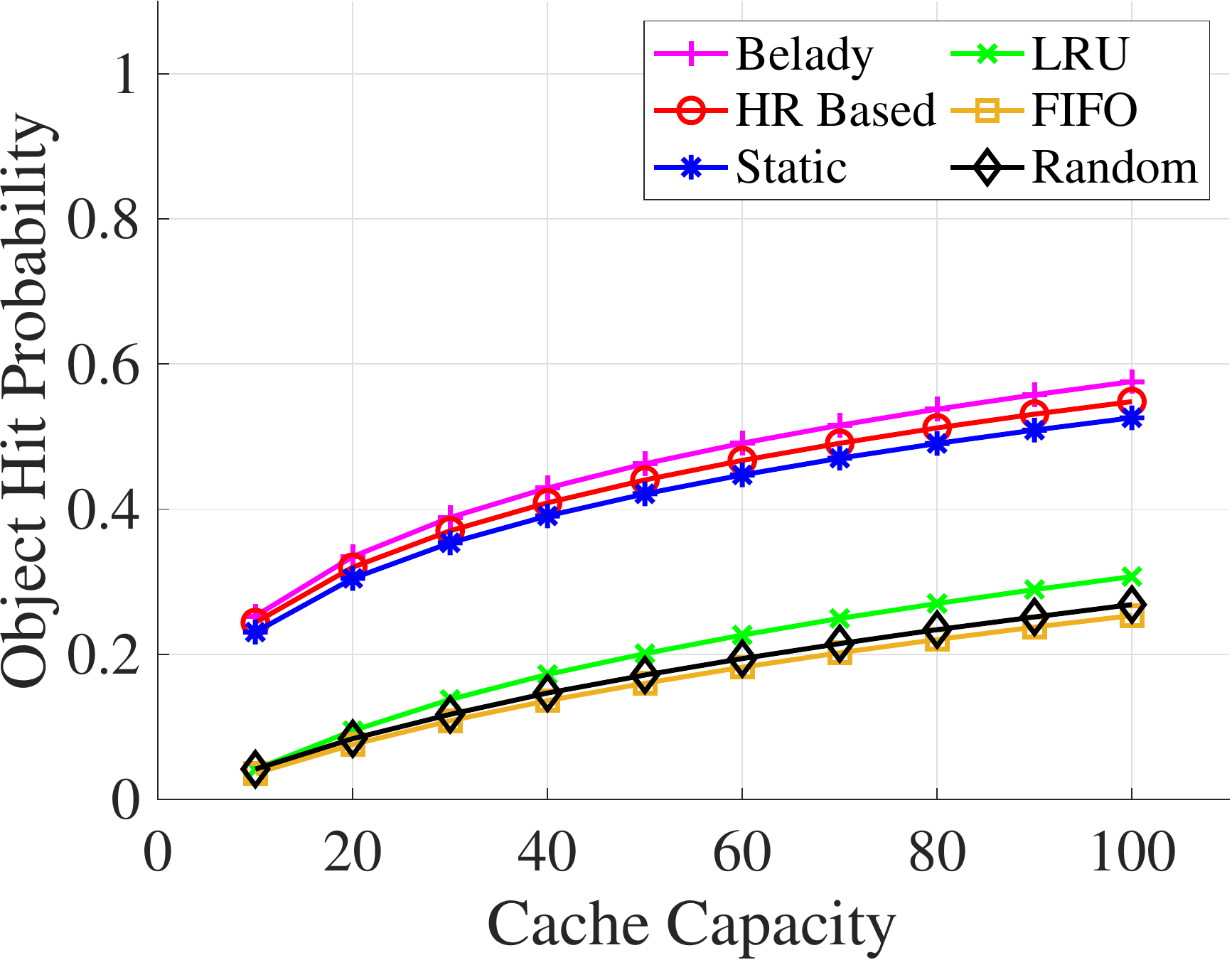}
\subcaption{Erlang (IHR)}
\end{minipage}
\vspace{-0.1in}
\caption{Simulation results for HR based upper bound and various caching policies under different inter request arrival distributions ($n=1000$, all objects of size $1$).}
\label{plots}
\end{figure*}

As a general comment, we note from Figure \ref{plots}-\ref{real-trace} that, as expected, that the HR based upper bound  serves as an upper bound on the 
hit probability among all online caching policies.  Further comments are given below on each figure.

\subsubsection{Figure \ref{plots}: Renewal request processes and equal size objects}
\label{ssec:fig2}

Request processes used to generate plots in Figure \ref{plots} are presented in Section \ref{ssec:arrival}.  These plots are obtained for $1000$ objects and all objects have size $1$.
Notice (see discussion in Section \ref{ssec:Poisson}) that results coincide in Figure \ref{plots}(a) for the STATIC policy and the HR based upper bound  when inter-request times (IRTs) are exponential distributed. We observe that when IRTs are either CHR or DHR, the HR based upper bound is much tighter than B\'el\'ady's upper bound and that  both bounds are close
when IRTs are IHR. STATIC consistently yields the highest hit probability and is always close to the HR upper bound. For exponential IRTs or, equivalently for the independence reference model, the optimality of STATIC  is well known  \cite{liu98}.
\subsubsection{Figure \ref{plots-var}: Renewal request processes and variable size objects}
\label{ssec:fig3}
Request processes used to generate plots in Figure \ref{plots-var} are presented in Section \ref{ssec:arrival}.   Objects have variable  sizes (see Section \ref{ssec:size})
and there are $100$ objects. We observe that when IRTs have exponential or Generalized Pareto distributions  the HR based upper bound is much tighter than the FOO upper bound 
(Figure \ref{plots-var}(a)-(b)) and that  both bounds are close  when IRTs are uniformly distributed rvs (Figure \ref{plots-var}(c)).  For exponential and Generalized Pareto IRT distributions
the GDSF policy performs well (close to HR); one way of interpreting the gap between HR (resp. FOO) and
GDSF in Figure \ref{plots-var}(c) is to say that there is room for improvement in caching policy performance when IRTs are uniformly distributed rvs.


\begin{figure*}[h]
\vspace{-0.1in}
\centering
\hspace{-0.3cm}
\begin{minipage}{0.32\textwidth}
\includegraphics[width=1\textwidth]{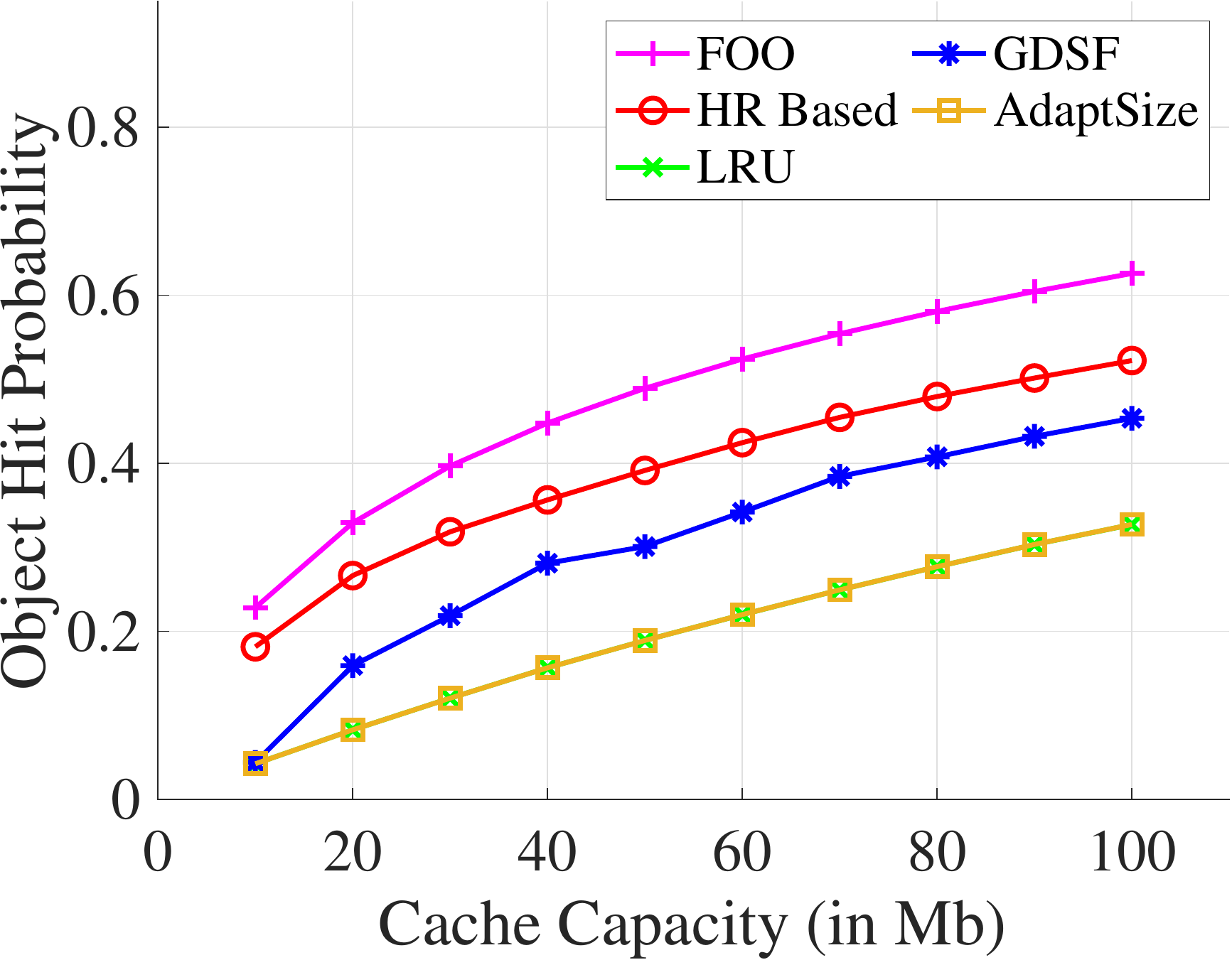}
\subcaption{Exponential (CHR)}
\end{minipage}
\begin{minipage}{0.32\textwidth}
\includegraphics[width=1\textwidth]{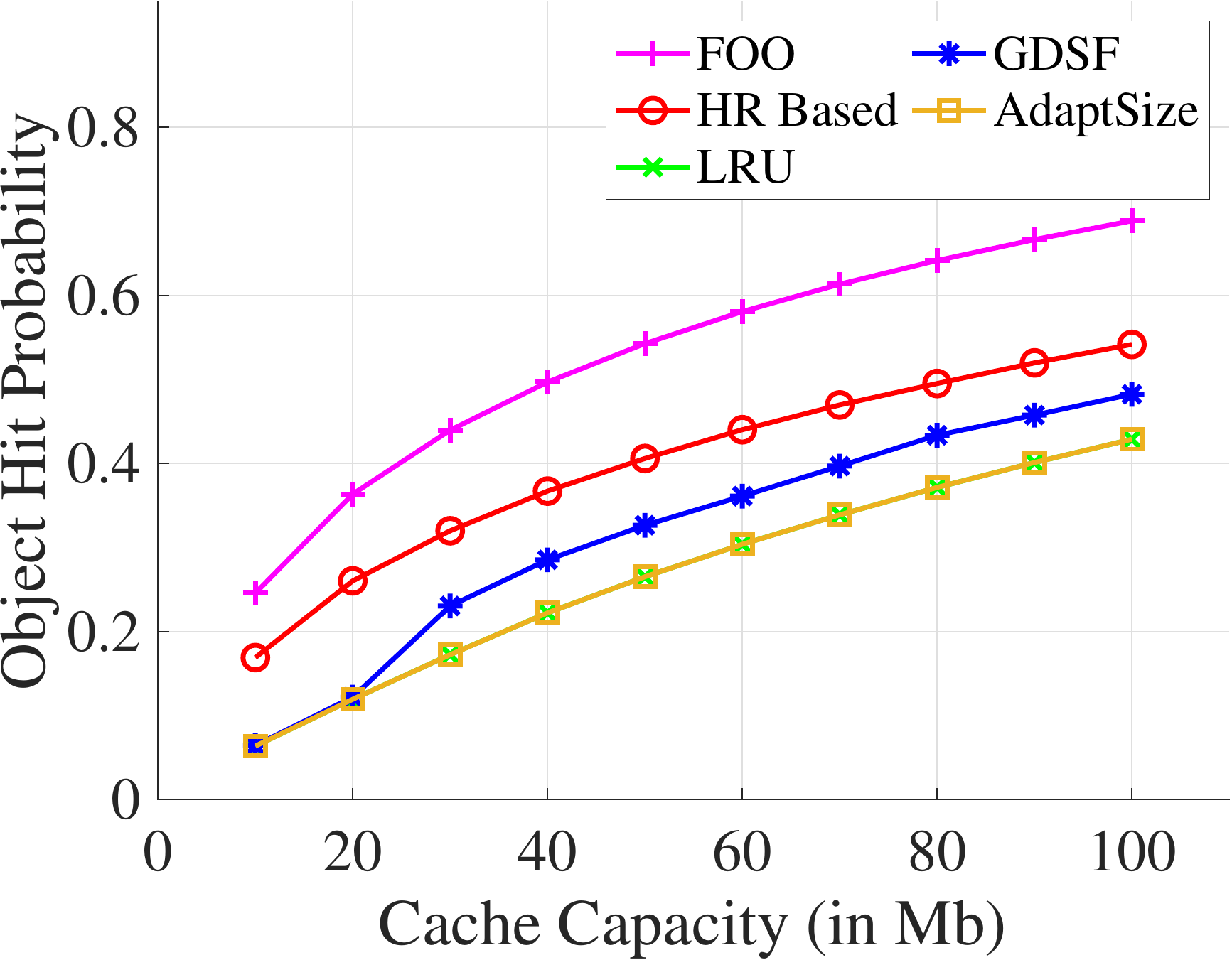}
\subcaption{Generalized Pareto (DHR)}
\end{minipage}
\begin{minipage}{0.32\textwidth}
\includegraphics[width=1\textwidth]{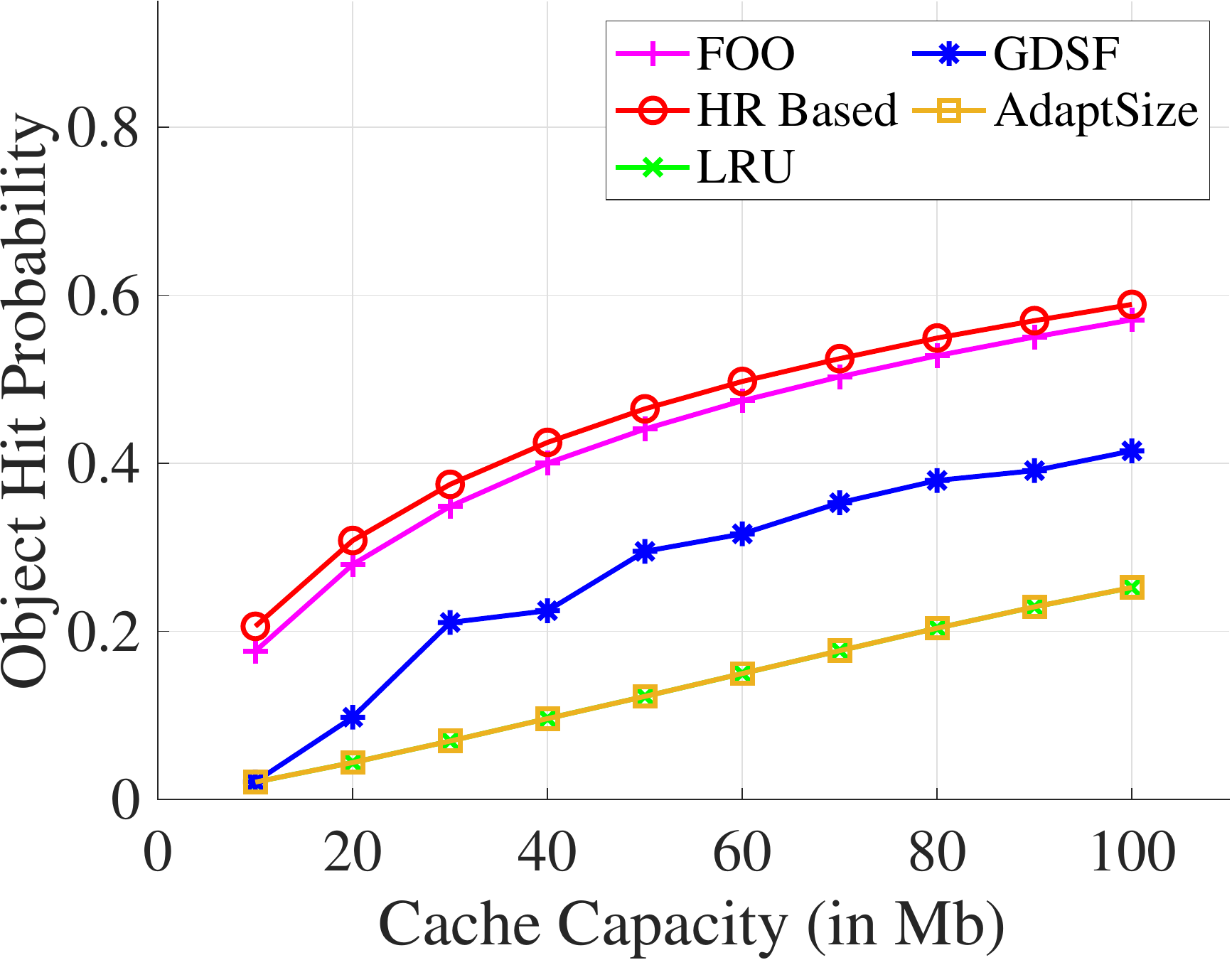}
\subcaption{Uniform (IHR)}
\label{pltuni}
\end{minipage}
\caption{Simulation Results for HR based upper bound and various caching policies under different inter request arrival distributions for variable object sizes ($n=100$).}
\label{plots-var}
\end{figure*}
\vspace{-0.4in}
\begin{figure}[h]
\vspace{-0.8in}
\centering
\begin{minipage}{0.55\textwidth}
\includegraphics[width=1\textwidth]{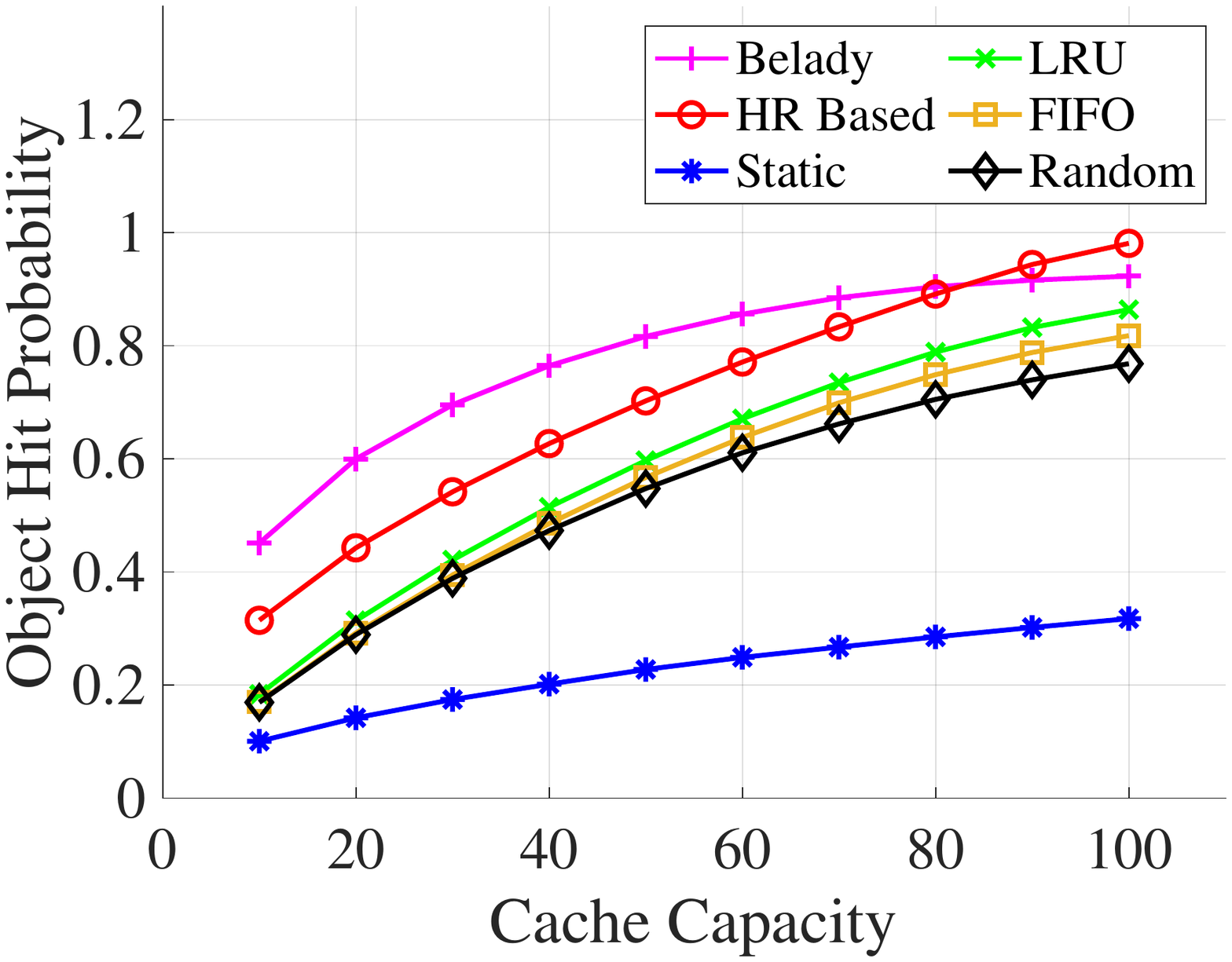}
\vspace*{-2.3cm}
\subcaption{Equal Size ($n=1000$)}
\end{minipage}
\begin{minipage}{0.55\textwidth}
\includegraphics[width=1\textwidth]{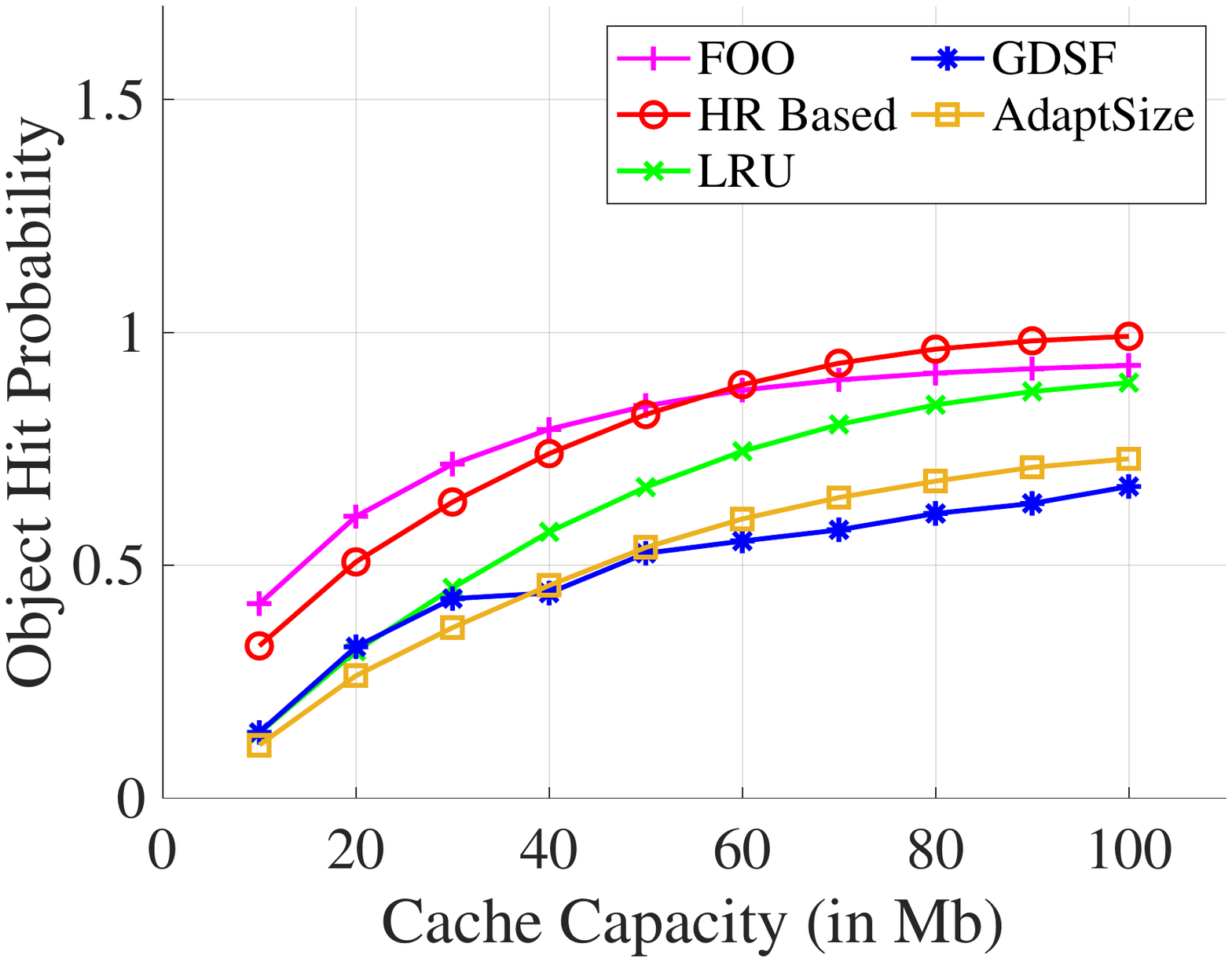}
\vspace*{-2.3cm}
\subcaption{Variable Size ($n=100$)}
\end{minipage}
\caption{Performance comparison under on-off request process}
\label{on-off-trace}
\end{figure}

\subsubsection{Figure \ref{on-off-trace}: on-off request arrivals and equal/variable size objects}
\label{ssec:fig4}
The parameters of the on-off process used to generate arrival times of requests of object $i$ ($i=1,\ldots,n$) are given in Section \ref{ssec:arrival}. 
There are $1000$ objects in the catalog for equal sized objects and $100$ objects for variable sized objects.
The average arrival rate for object $i$ is $\lambda_i \pi_{i,1}$ with $\pi_{i,1}=\alpha_i/(\alpha_i+\beta_i)$  (Section \ref{ssec:onoff}). The STATIC policy permanently stores
in the cache the $B$ objects in decreasing order of $\{\lambda_i \pi_{i,1}\}_i$.

For equal sized objects (resp. variable size objects) the HR bound is tighter than BELADY (resp. FOO) for low cache 
sizes whereas for higher cache sizes, BELADY (resp. FOO) becomes tighter. LRU performs the best for both equal sized and variable sized objects and 
STATIC policy performs the worst for equal sized objects. 

\subsubsection{Figure \ref{mmpp}: MMPP request arrivals and equal/variable size objects}
\label{ssec:fig5}
The parameters of the two-state MMPP (states $1$ and $2$) used to generate arrival times of requests are given in Section \ref{ssec:arrival}. There are $1000$ objects in the catalog
for equal sized objects and $100$ objects for variable sized objects.
The average arrival rates for object $i$ is  $\lambda_i(1)\gamma(1)+\lambda_i(2)\gamma(2)= (\lambda_i(1)\beta +  \lambda_i(2)\alpha)/(\alpha+\beta)$.
The STATIC caching policy permanently stores in the cache the $B$ objects with the highest average arrival rates.

Unlike in Figures \ref{plots}-\ref{on-off-trace}, BELADY is tighter than the HR based upper bound for equal sized objects (Figure \ref{mmpp}(a)) but the latter upper bound is
tighter than the FOO upper bound for variable sized objects (Figure \ref{mmpp}(b)).  STATIC performs the best among all  online caching policies. Note that, in our simulations, 
$\gamma(1)=\beta/(\alpha+\beta)$ and $\gamma(2)=\alpha/(\alpha+\beta)$ are comparable. We postulate that  the performance of STATIC will further improve when $\gamma(1) \gg\gamma(2)$  or $\gamma(1) \ll \gamma(2).$ For example, when $\gamma(1) \gg\gamma(2), \lambda_i^{STATIC}\sim  \lambda_i(1)\gamma(1)$; in this case the  STATIC policy will permanently store the popular objects in state $1$, thus always getting a hit when the MMPP is in state $1$.
\begin{figure*}[]
\vspace{-0.8in}
\centering
\begin{minipage}{0.55\textwidth}
\includegraphics[width=1\textwidth]{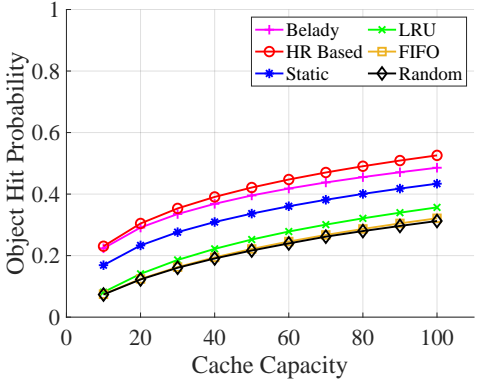}
\vspace*{-2.5cm}
\subcaption{Equal Size ($n=1000$)}
\end{minipage}
\begin{minipage}{0.55\textwidth}
\includegraphics[width=1\textwidth]{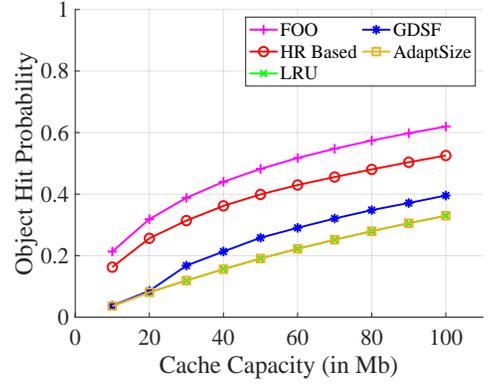}
\vspace*{-2.5cm}
\subcaption{Variable Size ($n=100$)}
\end{minipage}
\caption{Performance comparison under two-state MMPP request arrivals}
\label{mmpp}
\end{figure*}

\begin{figure}[]
\vspace{-0.8in}
\centering
\begin{minipage}{0.55\textwidth}
\includegraphics[width=1\textwidth]{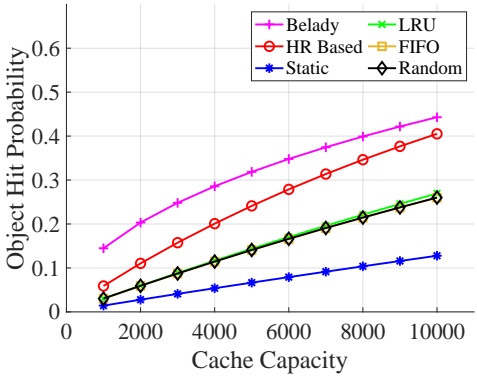}
\vspace*{-2.3cm}
\subcaption{Equal Size}
\end{minipage}
\begin{minipage}{0.55\textwidth}
\includegraphics[width=1\textwidth]{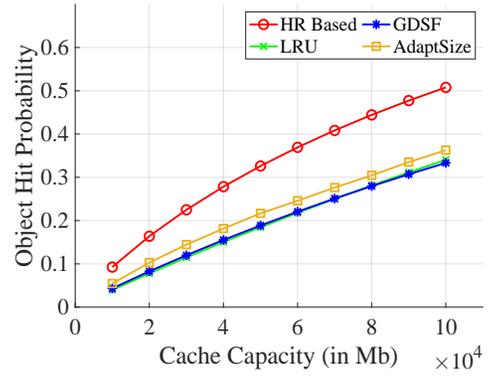}
\vspace*{-2.3cm}
\subcaption{Variable Size}
\end{minipage}
\caption{Performance comparison under shot noise model ($n=143871$).}
\label{shot-noise-trace}
\vspace{-0.15in}
\end{figure}
\subsubsection{Figure \ref{shot-noise-trace}: Shot noise request arrivals and equal/variable size objects}
\label{ssec:fig6}
The parameters of the  SNM used to generate Figure \ref{ssec:fig6} are given  in Section \ref{ssec:arrival}.
For equal sized objects (Figure \ref{ssec:fig6}(a)) our proposed HR bound not only upper bounds the hit probability for existing online caching policies but also provides a tighter bound than
the state-of-the-art BELADY.  STATIC policy performs the worst while LRU performs the best among all online policies. 
For variable sized objects ((Figure \ref{ssec:fig6}(b)) AdaptSize performs the best and  GDSF and LRU have similar performance. The difference in the object hit probability between the HR upper bound and AdaptSize suggests that there is room for improvement in caching policy performance.
\subsubsection{Figure \ref{real-trace}: Real-world trace}
\label{ssec:fig7}

Characteristics of the real-world trace and its application to the production of Figure \ref{real-trace} are discussed in Section \ref{ssec:arrival}. Upper bounds on the object hit probability
obtained  with HR and BELADY are almost identical. LRU performs the best and STATIC the worst.

\begin{figure}[]
\centering
\vspace*{-3.3cm}
\includegraphics[width=0.7\linewidth]{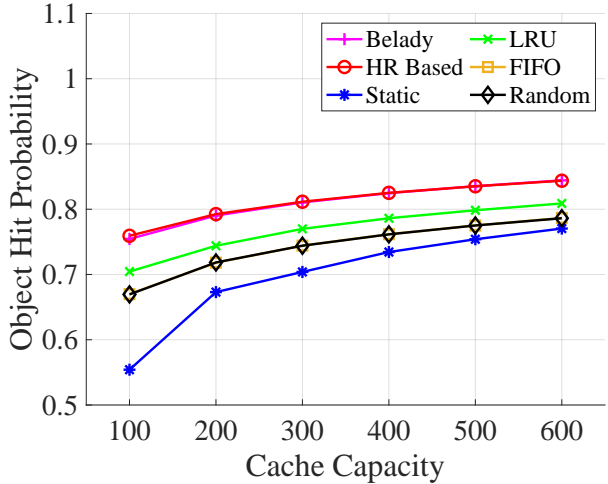}
\vspace*{-3.3cm}
\caption{Performance comparison under real world data trace ($n=5638$).}
\label{real-trace}
\vspace{-0.1in}
\end{figure}

\section{Related Literature}\label{sec:rel}

Many previous work has focused on improving cache hit probability for equal size objects \cite{Jiang2002, Megiddo2003, Tanenbaum01, Arlitt2000, Cao1997, Jaleel2010}. Also, many other policies have been proposed for variable size objects, for example: \emph{Least Recently Used} (LRU), \emph{Greedy Dual Size Frequency Caching Policy} (GDSF) \cite{GDSF98} and AdaptSize \cite{Berger17}. The primary objective of these policies is to improve object hit probability as opposed to the byte hit probability.
The optimal policy for equal size objects was first proposed by Belady et al. \cite{aho71}. Belady's algorithm uses exact knowledge of future requests, hence an offline policy.  When object requests follow the IRM, LFU achieves the maximum hit probability for equal sized objects \cite{liu98}. Computing optimal policy for variable object sizes is known to be NP-hard \cite{Chrobak12}.
Belady's algorithm serves as an upper bound on cache hit probability among all online demand based caching policies for equal size objects. For variable sized objects, upper bounds on object hit probability have also been proposed. Examples include Infinite-Cap \cite{Abrams95}, Flow-based offline optimal (FOO) and Practical Flow-based offline optimal (PFOO) \cite{berger18} policies. One major issue with all previous work is that the proposed upper bounds are offline, i.e. they assume exact knowledge of future requests. In this work, we proposed an online upper bound for both equal and variable sized objects with limited knowledge on object arrival process.
\section{Conclusion}\label{sec:con}
In this paper, we developed an upper bound on the cache hit probability for non-anticipative caching policies with equal object sizes. We showed that hazard rate associated with the object arrival process can be used to provide this upper bound. Inspired by the results for equal size objects, we extended the HR based argument to obtain an upper bound on the byte and object hit probability for variable size objects solving a knapsack problem. We derived closed form expressions for the upper bound under some specific object request arrival processes. 
We showed that HR based upper bound is tighter for a variety of object arrival processes than those analyzed in the literature. Future directions include to analyze the prefetching cost associated with any realizable hazard rate based caching policy. 

\bibliographystyle{ACM-Reference-Format}
\bibliography{refs}  

\section{Appendix}\label{sec:app}

\subsection{Proof of Equation \eqref{pi-t}}\label{sub-app-pi}
%

We drop the argument $t$ in $k_i(t)$ (see definition  at the beginning of Section \ref{ssec:number-hits}  as no confusion may occur.).
%
We have
$$p_i(t)=\Prob \left( T_{i,k_i}< T_{j,k_j},  \forall j\not=i\,|\, \mathcal{H}_{t}, \min_{j=1,\ldots,n}   T_{j,k_j}=t\right).$$
For $h>0$
\begin{align}
&\Prob\bigg(T_{i,k_i}\in (t,t+h), T_{j,k_j}>t+h, \forall  j\not=i\,\bigg|\mathcal{H}_{t},\min_{j=1,\ldots,n}   T_{j,k_j}\in (t, t+h)\bigg)\nonumber\\
&= \frac{\Prob\left(T_{i,k_i}\in (t,t+h), T_{j,k_j}>t+h, \forall  j\not=i  \,|\, \mathcal{H}_{t}\right)}{
\Prob\left(\min_{j=1,\ldots,n}   T_{j,k_j}\in (t, t+h)\,|\, \mathcal{H}_{t}\right)} =\frac{\Prob (T_{i,k_i}\in (t,h)\,|\,  \mathcal{H}_{i,t} )
 \prod_{1\leq j\leq n\atop j\not=i} \Prob(T_{j,k_j}>t+h\,|\, \mathcal{H}_{j,t})}
{\Prob\left(\min_{j=1,\ldots,n}   T_{j,k_j}\in (t, t+h)\,|\, \mathcal{H}_{t}\right)},\label{int1}
\end{align}
from the conditional independence assumption in (\ref{assumption-arrivals}). 
Let us focus on the denominator in (\ref{int1}).  It can be written as
\begin{align}
\Prob\left(\min_{j=1,\ldots,n}   T_{j,k_j}\in (t, t+h)\,|\, \mathcal{H}_{t}\right)= \sum_{j=1}^n \Prob( T_{j,k_j}\in (t,t+h)>t+h, \forall l\not=j\,|\, H_{t}) + f(h),\nonumber\\
\end{align}
with $f(h)\to 0$ as $h\to 0$,  since as $h\to 0$ there can be at least one random variables (rvs) located in  $(t,t+h)$ among the rvs $T_{1,k_1}, \ldots, T_{n,k_n}$ since these rvs are continous.
Therefore,
\begin{align}
&\Prob \bigg(T_{i,k_i}\in (t,t+h), T_{j,k_j}>t+h, \forall  j \not = i\bigg|\, \mathcal{H}_{t},\min_{j=1,\ldots,n}   T_{j,k_j}\in (t, t+h)\bigg)\nonumber\\
&= \frac{\Prob (T_{i,k_i}\in (t,h)\,|\,  \mathcal{H}_{i,t} ) \prod_{1\leq j\leq n\atop j\not=i}  \Prob(T_{j,k_j}>t+h\,|\, \mathcal{H}_{j,t})}{ \sum_{j=1}^n \Prob( T_{j,k_j}\in (t,d+t), T_{l,k_l}>t+h, \forall l\not=j
\,|\, H_{t}) +f(h)}\nonumber \\\
&=\Prob (T_{i,k_i}\in (t,h)\,|\,  \mathcal{H}_{i,t} ) \prod_{1\leq j\leq n\atop j\not=i}  \Prob(T_{j,k_j}>t+h\,|\, \mathcal{H}_{j,t})\nonumber\\
&\times\Biggl( \sum_{j=1}^n \Prob( T_{j,k_j}  \in (t,t+h)  \,|\,  \mathcal{H}_{l,t})\times  \prod_{1\leq l\leq n\atop l\not=j}  \Prob( T_{l,k_l}>t+h\,|\, \mathcal{H}_{l,t})+f(h)\Biggr)^{-1}\label{eq1-assump} \\
&=\frac{\frac{\Prob (T_{i,k_i}\in (t,t+h)\,|\,  \mathcal{H}_{i,t})}{\Prob(T_{i,k_i}>t+h\,|\, \mathcal{H}_{i,t})}}
{\sum_{j=1}^n \frac{\Prob (T_{j,k_j}\in (t,t+h)\,|\,  \mathcal{H}_{l,t})}{\Prob(T_{j,k_j}>t+h\,|\, \mathcal{H}_{l,t})} +f(h)}=\frac{\lambda^*_i(t+h)}{\sum_{j=1}^n \lambda^*_j(t+h)+f(h)},
\label{t+h}
\end{align}
%
where (\ref{eq1-assump}) follows from (\ref{assumption-arrivals}) and (\ref{t+h}) follows from (\ref{hazard-rate}).
Letting $h\to 0$ in (\ref{t+h}) gives \eqref{pi-t}.

\subsection{Proof of Lemma \ref{lemma:mmpp}}\label{sub-app-mmpp-lemma}
Let $k\geq 1$. By definition of policy ${\rm HR}$
\begin{equation}
\sum_{i\in B_k^{\rm HR}} \lambda_i(X(T_k-))\geq \sum_{i\in B_k^\pi} \lambda_i(X(T_k-)),
\label{inq}
\end{equation}
since under ${\rm HR}$ the $B$ objects with the smallest $\lambda_i(X(T_k-))$ are in the cache at time $T_k$. We have
\begin{eqnarray*}
\E[H^{\rm HR}_k\,|\, T_k=t, X(T_k-)=x] 
&=& \frac{\sum_{i\in B_k^{\rm HR}}\lambda_i(x)}{ \sum_{j=1}^n \lambda_j(x)} \\
&\geq &  \frac{\sum_{i\in B_k^\pi}\lambda_i(x)}{ \sum_{j=1}^n \lambda_j(x)} \quad \hbox{from } (\ref{inq})\\
&=&\E[H^\pi_k\,|\, T_k=t, X(T_k-)=x] .
\end{eqnarray*}
Removing the conditioning on $T_k$ and $x$ proves that 
\[
\E[H^{\rm HR}_k]\geq \E[H^\pi_k].
\]
Hence,
\[
\E[N^{\rm HR}_K]=\sum_{k=1}^K \E[H^{\rm HR}_k] \geq \sum_{k=1}^K \E[H^\pi_k]
=\E[N^\pi_K].
\]

\end{document}